\documentclass[journal]{IEEEtran}
\usepackage{graphicx}
\usepackage{amsmath}
\usepackage{amssymb}
\usepackage{amsthm}
\usepackage[noend]{algpseudocode}
\usepackage{epstopdf}
\usepackage{algorithm}
\usepackage{color}
\usepackage{cite}
\usepackage{comment}
\usepackage{subcaption}
\usepackage{multirow}
\usepackage{url}
\usepackage{tabu}

\newtheorem{theorem}{Theorem}
\newtheorem{prop}{Proposition}
\newtheorem{remark}{Remark}

\newtheorem{lemma}{Lemma}
\newtheorem{example}{Example}

\newcommand{\red}[1]{\textcolor{red}{#1}}

\begin{document}

\title{\red{Optimizing Timer-based  Policies for General Cache Networks}} 
\title{A TTL-based Approach for Content Placement in Edge Networks}
\author{Nitish K. Panigrahy$^{*}$, Jian Li$^{*}$,~\IEEEmembership{Member,~IEEE,} Faheem Zafari,~\IEEEmembership{Student Member,~IEEE,} \\Don Towsley,~\IEEEmembership{Life Fellow,~IEEE,}  and Paul Yu,~\IEEEmembership{Member,~IEEE}
\thanks{${}^*$Authors with equal contribution.}\thanks{This research was sponsored by the U.S. ARL and the U.K. MoD under Agreement Number W911NF-16-3-0001 and by the NSF under Grant CNS-1617437. The views and conclusions contained in this document are those of the authors and should not be interpreted as representing the official policies, either expressed or implied, of the National Science Foundation, U.S. Army Research Laboratory, the U.S. Government, the U.K. Ministry of Defence or the U.K. Government. The U.S. and U.K. Governments are authorized to reproduce and distribute reprints for Government purposes notwithstanding any copyright notation hereon. } 
\thanks{N. Panigrahy and D. Towsley are with the University of Massachusetts, Amherst, MA, 01003, USA (e-mail: \{nitish, towsley\}@cs.umass.edu).}
\thanks{J. Li is with Binghamton University, the State University of New York, Binghamton, NY, 13902, USA (e-mail: lij@binghamton.edu).}
\thanks{F. Zafari is with Imperial College London, London, SW72BT, UK, (e-mail: faheem16@imperial.ac.uk).}
\thanks{P. Yu is with U.S. Army Research Laboratory, Adelphi, MD 20783, USA (e-mail: paul.l.yu.civ@mail.mil).}
}
\maketitle

\begin{abstract}

Edge networks are promising to provide better services to users by provisioning computing and storage resources at the edge of networks.   However, due to the uncertainty and diversity of user interests, content popularity, distributed network structure, cache sizes, it is challenging to decide where to place the content, and how long it should be cached.  In this paper,  we study the utility optimization of content placement at edge networks through timer-based (TTL) policies.  We propose provably optimal distributed algorithms that operate at each network cache to maximize the overall network utility.  Our TTL-based optimization model provides theoretical answers to how long each content must be cached, and where it should be placed in the edge network.  Extensive evaluations show that our algorithm significantly outperforms path replication with conventional caching algorithms over some network topologies.

%Caching algorithms are usually described by the eviction method and analyzed using a metric of hit probability.  Since contents have different importance (e.g. popularity), the utility of a high hit probability, and the cost of transmission can vary across contents.   In this paper, we consider timer-based (TTL) policies across a cache network, where contents have differentiated timers over which we optimize.  Each content is associated with a utility measured in terms of the corresponding hit probability.  We start our analysis from a linear cache network: we propose a utility maximization problem where the objective is to maximize the sum of utilities and a cost minimization problem where the objective is to minimize the content transmission cost across the network.  These frameworks enable us to design online algorithms for cache management, for which we prove achieving optimal performance.  Informed by the results of our analysis, we formulate a non-convex optimization problem for a general cache network.  We show that the duality gap is zero, hence we can develop a distributed iterative primal-dual algorithm for content management in the network.  Numerical evaluations show that  our algorithm significant outperforms path replication with traditional caching algorithms over some network topologies.  Finally, we consider a direct application of our cache network model to content distribution. 
\end{abstract}


\IEEEpeerreviewmaketitle

\section{Introduction}\label{sec:intro}

Content distribution has become a dominant application in today's Internet.  Much of these contents are delivered by Content Distribution Networks (CDNs), which are provided by Akamai, Amazon, etc \cite{jiang12}.  There usually exists a stringent requirement on the latency between the service provider and end users for these applications.  CDNs use a large network of caches to deliver content from a location close to the end users.  This aligns with the trend of edge networks, where computing and storage resources are provisioned at the edge of networks.  If a user's request is served by a nearby edge cache,  %(i.e., cache hit) 
 the user experiences a faster response time than if it was served by the backend server.   It also reduces bandwidth requirements at the central content repository.

With the aggressive increase in Internet traffic over past years \cite{cisco15}, CDNs need to host content from thousands of regions belonging to web sites of thousands of content providers.  Furthermore, each content provider may host a large variety of content, including videos, music, and images.  Such an increasing diversity in content services requires CDNs to provide different quality of service to varying content classes and applications with different access characteristics and performance requirements.  Significant economic benefits and important technical gains have been observed with the deployment of service differentiation \cite{feldman02}.  While a rich literature has studied the design of fair and efficient caching algorithms for content distribution, little work has paid attention to the provision of multi-level services in edge networks.

Managing edge networks requires policies to route end-user requests to the local distributed caches, as well as caching algorithms to ensure availability of requested content at the cache.  In general, there are two classes of policies for studying the performance of caching algorithms: \emph{timer-based, i.e., Time-To-Live (TTL)} \cite{fagin77,che02,fofack14} and  \emph{non-timer-based} caching algorithms, e.g., Least-Recently-Used (LRU) \cite{coffman73},  Least-Frequently-Used (LFU) \cite{coffman73}, First In First Out (FIFO), and RANDOM \cite{coffman73}. %LRU, FIFO and RANDOM. 
 Since the cache size is usually much smaller than the total amount of content, some contents need to be evicted if the requested content is not in the cache.  %(i.e., cache miss). 
%Some well known content caching algorithms are Least-Recently-Used (LRU) \cite{coffman73},  Least-Frequently-Used (LFU) \cite{coffman73}, First In First Out (FIFO), and RANDOM \cite{coffman73}.  
Exact analysis of LRU, LFU, FIFO and RANDOM has proven to be difficult, even under the simple \emph{Independence Reference Model} (IRM) \cite{coffman73}, where requests are independent of each other.   The strongly coupled nature of these eviction algorithms makes implementation of differential services challenging.  In contrast, a TTL cache associates each content with a timer upon request and the content is evicted from the cache on timer expiry, independent of other contents.  Analysis of these policies is simple since the eviction of contents are decoupled from each other.

Most studies have focused on the analysis of a single edge cache.  When an edge network is considered, independence across different caches is usually assumed \cite{rosensweig10}.   Again, it is hard to analyze  most conventional caching algorithms, such as LRU, FIFO and RANDOM,  but some accurate results for TTL caches are available \cite{berger14,fofack14}.  However, it has been observed \cite{laoutaris04} that performance gains can be obtained if decision-making is coupled at different caches.

\begin{figure}
\centering
\includegraphics[width=0.9\linewidth]{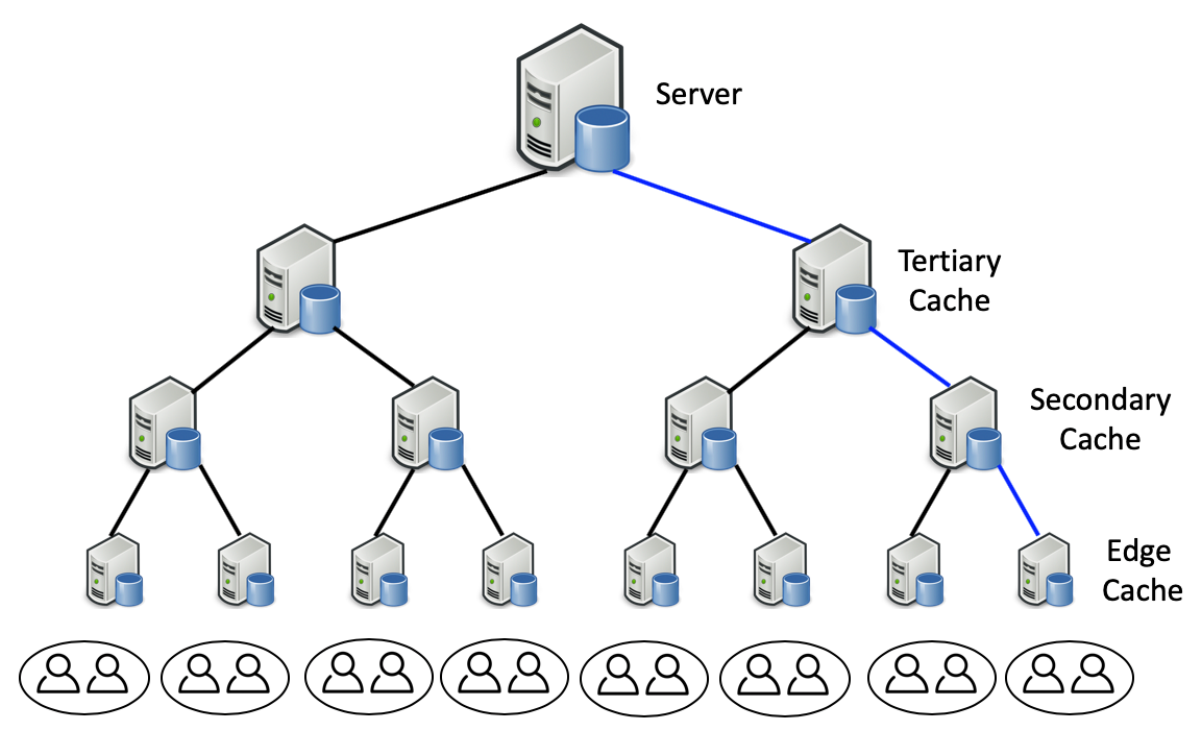}
\caption{An edge network with a server holding all contents and three layers of caches.  Each edge cache serves a set of users with different requests.  The blue line illustrates a unique path between the edge cache and the server.}
\label{fig:model}
\vspace{-0.1in}
\end{figure}

In this paper, we consider a TTL-based edge network, where a set of network caches host a library of unique contents, and serve a set of users.  Figure~\ref{fig:model} illustrates an example of such a network, which is consistent with the YouTube video delivery system \cite{ramadan17,sasikumar19}.  Each user can generate a request for a content, which is forwarded along a fixed \emph{path} from the edge cache towards the server.  Forwarding stops upon a cache hit, i.e., the requested content is found in a cache on the path.  When such a cache hit occurs, the content is sent over the reverse path to the edge cache initializing the request.  This raises the questions: \emph{where to cache the requested content on the reverse path} and \emph{what is the value of its timer?}  Answering these questions can provide new insights in edge network design.  However, it may also increase the complexity and hardness of the analysis.

Our goal is to provide thorough and rigorous answers to these questions.  To that end, we consider moving the content one cache up (towards the user) if there is a cache hit on it and pushing the content one cache down (away from the user) once its timer expires in the cache hierarchy, since the recently evicted content may still be in demand.  This policy is known as ``Move Copy Down with Push" (MCDP) policy.  %While pushing a copy down may improve system performance, it induces greater operational cost in the system. We can also consider another policy ``Move Copy Down" (MCD) under which content is evicted upon timer expiry. These will be described in detail in Section~\ref{sec:prelim}.

We first formulate a utility-driven caching framework for linear edge networks, where each content is associated with a utility and content is managed with a timer whose duration is set to maximize the aggregate utility for all contents over the edge network.  Building on MCDP policy, we formulate the optimal TTL policy as a non-convex optimization problem in Section~\ref{sec:line}.  One contribution of this paper is to show that this non-convex problem can be transformed into a convex one by change of variables.  We further develop distributed algorithms for content management over linear edge networks, and  show that this algorithm converges to the optimal solution.

Informed by our results for linear edge networks, we consider a general cache network where each edge cache serves distinct contents, i.e., there are no common contents among edge caches, in Section~\ref{sec:non-common}. We show that this edge network can be treated as a union of different linear edge networks between each edge cache and the central server. 

We further consider a more general case where common content is requested among edge caches in Section~\ref{sec:common}.  %The more interesting case where common content is requested along different paths is considered in Section~\ref{sec:general-cache-common}.  
This introduces non-convex constraints, resulting in a non-convex utility maximization problem.  We show that although the original problem is non-convex, the duality gap is zero.  Based on this, we design a distributed iterative primal-dual algorithm for content placement in edge network.   We show through numerical evaluations that our algorithm outperforms path replication with traditional caching algorithms over a broad array of network topologies. % \np{[np: ``significantly'' would be an over statement and we did not compare MCDP with other policies for topologies other than line or tree.]}

We provide some discussions on extension of our utility maximization framework under MCDP to a general graph based cache networks in Section~\ref{sec:app}.  We show how our framework can be directly mapped to content distributions in CDNs, ICNs/CCNs etc.  Numerical results are given on how to optimize the performance.  Also, since we consider an edge network, content placement induces costs, such as \emph{search cost} for finding the requested content on the path, \emph{fetch cost} to serve the content to the user that requested it, and \emph{move cost} upon cache hit or miss due to caching policy.  We fully characterize these costs and formulate a cost minimization problem in Section~\ref{sec:cost}.   We discuss related works in Section~\ref{sec:related} and conclude the paper in Section~\ref{sec:conclusion}.

%\blue{Utilities characterize user satisfaction and provide an implicit notion of fairness.  However,  since we consider a cache network, content management also induces costs, such as \emph{search cost} for finding the requested content on the path, \emph{fetch cost} to serve the content to the user that requested it, and \emph{move cost} upon cache hit or miss due to caching policy.  We fully characterize these costs and formulate a cost minimization problem in Section~\ref{sec:line-cost-min}. }

%\blue{Finally, we include some generalization in Section~\ref{sec:app}. We discuss how our framework can be directly mapped to content distributions in CDNs, ICNs/CCNs etc.  Numerical results are given on how to optimize the performance.  Conclusions are given in Section~\ref{sec:conclusion}.  Some additional discussions and proofs are provided in Appendix~\ref{appendix}.}

\section{Preliminaries}\label{sec:prelim}

%We consider an edge network, represented by a graph $G=(V, E),$ where a set of network caches host a library of $n$ unique contents, and serve a set of users.  \blue{Figure~\ref{} illustrates an example of such a network, which ...}  
We represent the edge cache network shown in Figure~\ref{fig:model} by a graph $G=(V, E).$  We use $\mathcal{D}=\{d_1, \cdots, d_n\}$ with $|\mathcal{D}|=n$ to denote the set of contents.  Each network cache $v\in V$ can store up to $B_v$ contents to serve requests from users.   We assume that each user will first send a request for the content to its local network cache, which may then route the request to other caches for retrieving the content.  Without loss of generality, we assume that there is a fixed and unique path from the local cache towards a terminal cache that is connected to a server that always contains the content.   

To be more specific, a request $(v, i, p)$ is determined by the local cache, $v$, that firstly received the user request, the requested content, $i$, and the path, $p$, over which the request is routed.  We denote a path $p$ of length $|p|=L$ as a sequence $\{v_{1p}, v_{2p}, \cdots, v_{Lp}\}$ of nodes $v_{lp}\in V$ such that $(v_{lp}, v_{(l+1)p})\in E$ for $l\in\{1, \cdots, L\},$ where $v_{Lp}=v.$   We assume that path $p$ is loop-free and terminal cache $v_{1p}$ is the only cache on path $p$ that accesses the server for content $i.$

%We assume a library of $n$ unique contents, denoted as $\mathcal{D}=\{d_1, \cdots, d_n\}$ with $|\mathcal{D}|=n.$ Each node can store a finite number of contents, $B_v$ is the cache capacity at node $v\in V.$  The network serves content requests routed over the graph $G.$  A request is determined by the item requested by the user and the path that the request follows; this will be described in detail in Section~\ref{sec:content-request}.  

We assume that the request processes for distinct contents are described by independent Poisson processes with arrival rate $\lambda_i$ for content $i\in\mathcal{D}.$ Denote $\Lambda = \sum_{i=1}^n \lambda_i.$  Then the popularity (request probability) of content $i$ satisfies \cite{baccelli13}
\begin{align}
\rho_i=\frac{\lambda_i}{\Lambda}, \quad i=1,\cdots, n.
\end{align}

We consider TTL caches in this paper.  Each content $i$ is associated with a timer $T_{ij}$ at cache $j$.  Suppose content $i$ is requested and routed along path $p.$ There are two cases: (i) content $i$ is not in any cache along path $p,$ in which case content $i$ is fetched from the server and inserted into the terminal cache (denoted by cache $1$)\footnote{Since we consider path $p$, for simplicity, we move the dependency on $p$ and $v$, denote it as nodes $1,\cdots, L$ directly.}. Its timer is set to $T_{i1}$; (ii) if content $i$ is in cache $l$ along path $p,$ content $i$ is moved to cache $l+1$ preceding cache $l$ in which $i$ is found, and the timer at cache $l+1$ is set to $T_{i{(l+1)}}$.  Content $i$ is pushed one cache back to cache $l-1$ and the timer is set to $T_{i(l-1)},$ once the timer expires.  We call it Move Copy Down with Push (MCDP) \cite{rodriguez16}.  Denote the \emph{hit probability} of content $i$ as $h_i,$ then the corresponding \emph{hit rate} is $\lambda_i h_i.$
%we consider the following strategies \cite{rodriguez16} 
%\begin{itemize}
%\item \red{{\textbf{Move Copy Down (MCD)}:} content $i$ is moved to cache $l+1$ preceding cache $l$ in which $i$ is found, and the timer at cache $l+1$ is set to $T_{i{(l+1)}}$.  Content $i$ is discarded once the timer expires;}
%\item {\textbf{Move Copy Down with Push (MCDP)}:} MCDP behaves the same as MCD upon a cache hit.  However, if timer $T_{il}$ expires, content $i$ is pushed one cache back to cache $l-1$ and the timer is set to $T_{i(l-1)}.$
%\end{itemize}

%\subsection{Utility Function} 
%\red{Utility functions capture the satisfaction perceived by a user after being served a content.   We associate each content $i\in\mathcal{D}$ with a utility function $U_i: [0,1]\rightarrow\mathbb{R}$ that is a function of hit probability $h_i$.  $U_i(\cdot)$ is assumed to be increasing, continuously differentiable, and strictly concave. In particular, for our numerical studies, we focus on the widely used $\beta$-fair utility functions \cite{srikant13} given by 
%\begin{equation}\label{eq:utility}
%    U_i(h)=
%\begin{cases}
%     w_i\frac{h^{1-\beta}}{1-\beta},& \beta\geq0, \beta\neq 1;\\
%     w_i\log h,& \beta=1,
%\end{cases}
%\end{equation}
%where $w_i>0$ denotes a weight associated with content $i$.}

%We consider the general cache network described in Section~\ref{sec:prelim}.  
Denote by $\mathcal{P}$ the set of all requests, and $\mathcal{P}^i$ the set of requests for content $i.$   Suppose a network cache $v$ serves two requests $(v_1, i_1, p_1)$ and $(v_2, i_2, p_2)$, then there are two cases: (i) non-common requested content, i.e., $i_1\neq i_2;$ and (ii) common requested content, i.e., $i_1=i_2.$  In the following, we will focus on how to design optimal TTL policies for content placement in an edge cache network under these two cases.% respectively. 

While classical cache eviction policies such as LRU provide good performance and are easy to implement, Garetto et al. \cite{garetto16} showed that K-LRU\footnote{{K-LRU adds $K-1$ meta-caches ahead of the real cache. Only ``popular" contents (requested at least $K-1$ times) are stored in real cache.}} can provide significant improvements over LRU even for very small K.  Furthermore, Ramadan et al. \cite{ramadan17} proposed K-LRU with big cache abstraction (K-LRU(B)) to effectively utilize resources in a hierarchical network of cache servers.  %In this work, we consider both K-LRU and K-LRU(B) for performance comparison.
Thus, in the rest of the paper, we compare the performance of MCDP to K-LRU and K-LRU(B).

%\np{
%Need to introduce hit rate and utility as a function of hit rate with references.
%}

\section{Linear Edge Network}\label{sec:line}

%\blue{Need to argue that G-N-U-MCDP can be broken into different subproblems.  Each subproblem is an optimization problem on a linear cache network}

We begin with a linear edge network, i.e., there is a single path between the user and the server,  composed of $|p|=L$ caches labeled $1,\cdots, L.$  A content enters the edge network via cache $1,$ and is promoted to a higher index cache whenever a cache hit occurs.  In the following, we consider the MCDP replication strategy when each cache operates with a TTL policy. We also consider another replication strategy ``Move Copy Down" (MCD) and relegate our discussion on it to Appendix \ref{mcdmodel-appendix} and Appendix \ref{app-mcd-opt}.

%We begin with a particular path %\footnote{We use the terms ``path'' and ``line'' interchangeably.} 
% $p$ between user and server,  composed of $|p|=L$ caches labeled $1,\cdots, L.$   A content enters the network via ``cache $1$", and is promoted to a higher index cache whenever a cache hit occurs.  We consider the MCDP replication strategy when each cache operates with a TTL policy.

\subsection{Stationary Behavior}\label{sec:ttl-stationary}

%\cite{gast16} considered a caching policy LRU($\boldsymbol m$). Though the policy differ from MCDP, the stationary analysis is similar.  We present our result here for completeness, which will be used subsequently in the paper.

%\subsubsection{MCDP}\label{mcdpmodel}
%\noindent\textbf{\textit{MCDP:} } 
Requests for content $i$ arrive according to a Poisson process with rate $\lambda_i.$  Under TTL, content $i$ spends a deterministic time in a cache if it is not requested, independent of all other contents. We denote the timer as $T_{il}$ for content $i$ in cache $l$ on the path $p,$ where $l\in\{1,\cdots, |p|\}.$  

Denote by $t_k^i$ the $k$-th time that content $i$ is either requested or the timer expires.  For simplicity, we assume that content is in cache $0$ (i.e., server) when it is not in the cache network.  We then define a discrete time Markov chain (DTMC) $\{X_k^i\}_{k\geq0}$ with $|p|+1$ states, where $X_k^i$ is the index of the cache that content $i$ is in at time $t_k^i.$  The event that the time between two requests for content $i$ exceeds $T_{il}$ occurs with probability $e^{-\lambda_i T_{il}}$; consequently we obtain the transition probability matrix of $\{X_k^i\}_{k\geq0}$ and compute the stationary distribution. %Details can be found in Appendix~\ref{mcdpmodel-appendix}.  
 The timer-average probability that content $i$ is in cache $l\in\{1,\cdots, |p|\}$ is 
\begin{subequations}\label{eq:hit-prob-mcdp}
\begin{align}
& h_{i1} = \frac{e^{\lambda_iT_{i1}}-1}{1+\sum_{j=1}^{|p|}(e^{\lambda_iT_{i1}}-1)\cdots (e^{\lambda_iT_{ij}}-1)},\label{eq:mcdp1}\\
& h_{il} = h_{i(l-1)}(e^{\lambda_iT_{il}}-1),\; l = 2,\cdots,|p|,\label{eq:mcdp2}
\end{align}
\end{subequations}
where $h_{il}$ is also the hit probability for content $i$ at cache $l.$
\begin{remark}
The stationary analysis of MCDP is similar to a different caching policy LRU($\boldsymbol m$) considered in \cite{gast16}.  We relegate its explicit expression to Appendix \ref{mcdpmodel-appendix}, and also refer interested readers to  \cite{gast16} for more detail. 
%\cite{gast16} considered a caching policy LRU($\boldsymbol m$). Though the policy differ from MCDP, the stationary analysis is similar, hence we omit the proof here due to space constraints.  
\end{remark}
%\subsubsection{MCD}
%\noindent\textbf{\textit{MCD:} } 
%\red{Again, under TTL, content $i$ spends a deterministic time $T_{il}$ in cache $l$ if it is not requested, independent of all other contents.  We define a DTMC $\{Y_k^i\}_{k\geq0}$ by observing the system at the time that content $i$ is requested. Similar to MCDP, if content $i$ is not in the cache network, it is in cache $0$; thus we still have $|p|+1$ states.  If $Y_k^i=l$, then the next request for content $i$ comes within time $T_{il}$ with probability $1-e^{-\lambda_iT_{il}}$, and $Y_{k+1}^i=l+1,$ otherwise $Y_{k+1}^i=0$ due to the MCD policy.  We can obtain the transition probability matrix of $\{Y_k^i\}_{k\geq0}$ and compute the stationary distribution, details are available in Appendix~\ref{mcdmodel-appendix}.} 
%
%
%\red{By the PASTA property \cite{MeyTwe09}, it follows that the stationary probability that content $i$ is in cache $l\in\{1,\cdots, |p|\}$ is \begin{subequations}\label{eq:hit-prob-mcd}
%\begin{align}
%&h_{il}=h_{i0}\prod_{j=1}^{l}(1-e^{-\lambda_iT_{ij}}),\quad l=1,\cdots, |p|-1,\displaybreak[1]\label{eq:mcd2}\\
%&h_{i|p|}=e^{\lambda_i T_{i|p|}}h_{i0}\prod_{j=1}^{|p|-1}(1-e^{-\lambda_iT_{ij}}),\label{eq:mcd3}
%\end{align}
%\end{subequations}
%where $h_{i0}=1/[1+\sum_{l=1}^{|p|-1}\prod_{j=1}^l(1-e^{-\lambda_i T_{ij}})+e^{\lambda_i T_{i|p|}}\prod_{j=1}^{|p|}(1-e^{-\lambda_i T_{ij}})].$}

\subsection{From Timer to Hit Probability}\label{sec:ttl-hit-prob}
We consider a TTL cache network where requests for different contents are independent of each other and each content $i$ is associated with a timer $T_{il}$ at each cache $l\in\{1,\cdots,|p|\}$ on the path.  Denote $\boldsymbol T_i=(T_{i1},\cdots, T_{i|p|})$ and $\boldsymbol T=(\boldsymbol T_1,\cdots, \boldsymbol T_n)$.  From~(\ref{eq:hit-prob-mcdp}), % and~(\ref{eq:hit-prob-mcd}), 
the overall utility on path $p$ is given as 
%\begin{small}
\begin{align}
\sum_{i\in\mathcal{D}}\sum_{l=1}^{|p|}\psi^{|p|-l}U_i(\lambda_ih_{il}(\boldsymbol T)),
\end{align}
%\end{small}
where the utility function $U_i: [0,\infty)\rightarrow\mathbb{R}$ is assumed to be increasing, continuously differentiable, and strictly concave function of content hit rate,  and $0<\psi\leq1$ is a discount factor capturing the utility degradation along the request's routing direction.   Since each cache is finite in size, we have the capacity constraint 
%\begin{small}
\begin{align}
\sum_{i\in\mathcal{D}}h_{il}(\boldsymbol T)\leq B_l, \quad l\in\{1,\cdots,|p|\}.
\end{align}
%\end{small}
Therefore, the optimal TTL policy for content placement on path $p$ is the solution of the following optimization problem
\begin{align}\label{eq:max-ttl}
\max_{\boldsymbol T} \quad&\sum_{i\in\mathcal{D}}\sum_{l=1}^{|p|}\psi^{|p|-l}U_i(\lambda_ih_{il}(\boldsymbol T))\nonumber\displaybreak[0]\\
\text{s.t.}\quad&\sum_{i\in\mathcal{D}}h_{il}(\boldsymbol T)\leq B_l, \quad l\in\{1,\cdots,|p|\},\nonumber\displaybreak[1]\\
&T_{il} \ge 0, \quad\forall i \in \mathcal{D}, \quad l=1, \cdots, |p|,
\end{align}
where $h_{il}(\boldsymbol T)$ is given in~(\ref{eq:hit-prob-mcdp}).  However, ~(\ref{eq:max-ttl}) is a non-convex optimization with a non-linear constraint.  Our objective is to characterize the optimal timers for different contents on path $p.$.  To that end, it is helpful to express~(\ref{eq:max-ttl}) in terms of hit probabilities. In the following, we discuss how to change the variables from timer to hit probability. 

%\subsubsection{MCDP}
%\noindent\textbf{\textit{MCDP:} } 
Since $0\leq T_{il}\leq \infty$, it is easy to check that $0\leq h_{il}\leq 1$ for $l\in\{1,\cdots,|p|\}$ from~(\ref{eq:mcdp1}) and~(\ref{eq:mcdp2}). Furthermore, it is clear that there exists a mapping between $(h_{i1},\cdots, h_{i|p|})$ and $(T_{i1},\cdots, T_{i|p|}).$ By simple algebra, we obtain
\begin{subequations}\label{eq:stationary-mcdp-timers}
\begin{align}
& T_{i1} = \frac{1}{\lambda_i}\log \bigg(1 + \frac{h_{i1}}{1-\big(h_{i1} + h_{i2} + \cdots + h_{i|p|}\big)}\bigg), \label{eq:mcdpttl1}\\
& T_{il} = \frac{1}{\lambda_i}\log \bigg(1 + \frac{h_{il}}{h_{i(l-1)}}\bigg),\quad l= 2,\cdots, |p|.\label{eq:mcdpttl2}
\end{align}
\end{subequations}
Note that 
%\begin{small}
\begin{align}\label{eq:mcdp-constraint}
h_{i1} + h_{i2} + \ldots + h_{i|p|} \leq 1,
\end{align}
%\end{small}
must hold during the operation, which is always true for our caching policies.

%\subsubsection{MCD}
%\noindent\textbf{\textit{MCD:} } 
%\red{Similarly, from~(\ref{eq:mcd2}) and~(\ref{eq:mcd3}), we simply check that there exists a mapping between $(h_{i1},\cdots, h_{i|p|})$ and $(T_{i1},\cdots, T_{i|p|}).$ Since $T_{il}\geq0,$ by~(\ref{eq:mcd2}), we have 
%\begin{align} \label{eq:mcd-constraint1}
%h_{i(|p|-1)} \leq h_{i(|p|-2)} \leq \cdots \leq h_{i1} \leq h_{i0}.
%\end{align}
%By simple algebra, we can obtain
%\begin{subequations}\label{eq:stationary-mcd-timers}
%\begin{align}
%& T_{i1} = -\frac{1}{\lambda_i}\log \bigg(1 - \frac{h_{i1}}{1-\big(h_{i1} + h_{i2} + \cdots + h_{i|p|}\big)}\bigg), \label{eq:mcdttl1}\\
%& T_{il} = -\frac{1}{\lambda_i}\log\bigg(1-\frac{h_{il}}{h_{i(l-1)}}\bigg),\quad l= 2, \cdots , |p|-1, \label{eq:mcdttl2}\\
%& T_{i|p|} = \frac{1}{\lambda_i}\log\bigg(1+\frac{h_{i|p|}}{h_{i(|p|-1)}}\bigg). \label{eq:mcdttl3}
%\end{align}
%\end{subequations}
%Again 
%\begin{align} \label{eq:mcd-constraint2}
%h_{i1} + h_{i2} + \cdots + h_{i|p|} \leq 1, 
%\end{align}
%must hold during the operation to obtain~(\ref{eq:stationary-mcd-timers}), which is always true for MCD.}

\subsection{Maximizing Aggregate Utility}\label{sec:line-utility-max}
With the change of variables discussed above, %given~(\ref{eq:stationary-mcdp-timers}) and~(\ref{eq:mcdp-constraint}), 
we can reformulate~(\ref{eq:max-ttl}) as follows %for MCDP and MCD, respectively. 
%\subsubsection{MCDP}\label{sec:line-utility-max-mcdp}
%\noindent\textbf{\textit{MCDP:} } 
%Given~(\ref{eq:stationary-mcdp-timers}) and~(\ref{eq:mcdp-constraint}), optimization problem~(\ref{eq:max-ttl}) under MCDP becomes
\begin{subequations}\label{eq:max-mcdp}
\begin{align}
%\text{\bf{L-U-MCDP:}} 
\max \quad&\sum_{i\in \mathcal{D}} \sum_{l=1}^{|p|} \psi^{|p|-l} U_i(\lambda_ih_{il}) \displaybreak[0]\\
\text{s.t.} \quad&\sum_{i\in \mathcal{D}} h_{il} \leq B_l,\quad l=1, \cdots, |p|,  \displaybreak[1]\label{eq:hpbmcdp1}\\
& \sum_{l=1}^{|p|}h_{il}\leq 1,\quad\forall i \in \mathcal{D},  \displaybreak[2]\label{eq:hpbmcdp2}\\
&0\leq h_{il}\leq1, \quad\forall i \in \mathcal{D}, \quad l=1, \cdots, |p|\label{eq:hpbmcdp3}
\end{align}
\end{subequations}
where~(\ref{eq:hpbmcdp1}) is the cache capacity constraint and~(\ref{eq:hpbmcdp2}) is due to the variable exchanges under MCDP as discussed above.%in~(\ref{eq:mcdp-constraint}) . 

\begin{prop}
Optimization problem defined in~(\ref{eq:max-mcdp}) under MCDP has a unique global optimum.
\end{prop}

%%\subsubsection{MCD}\label{sec:line-utility-max-mcd}
%\noindent\textbf{\textit{MCD:} } 
%\red{Given~(\ref{eq:mcd-constraint1}),~(\ref{eq:stationary-mcd-timers}) and~(\ref{eq:mcd-constraint2}), optimization problem~(\ref{eq:max-ttl}) under MCD becomes
%\begin{subequations}\label{eq:max-mcd}
%\begin{align}
%\text{\bf{L-U-MCD:}} \max \quad&\sum_{i\in \mathcal{D}} \sum_{l=1}^{|p|} \psi^{|p|-l} U_i(h_{il}) \displaybreak[0] \\
%\text{s.t.} \quad&\sum_{i\in \mathcal{D}} h_{il} \leq B_l,\quad l=1, \cdots, |p|,  \displaybreak[1]\label{eq:hpbmcd1}\\
%&h_{i(|p|-1)} \leq \cdots \leq h_{i1} \leq h_{i0},\quad\forall i \in \mathcal{D}, \displaybreak[2]\label{eq:hpbmcd2}\\
%& \sum_{l=1}^{|p|}h_{il}\leq 1,\quad\forall i \in \mathcal{D}, \displaybreak[3] \label{eq:hpbmcd3}\\
%&0\leq h_{il}\leq1, \quad\forall i \in \mathcal{D}, \label{eq:hpbmcdp4}
%\end{align}
%\end{subequations}
%where~(\ref{eq:hpbmcd1}) is the cache capacity constraint, ~(\ref{eq:hpbmcd2}) and~(\ref{eq:hpbmcd3}) are due to the variable exchanges under MCD as discussed in~(\ref{eq:mcd-constraint1}) and~(\ref{eq:mcd-constraint2}).}
%
%
%\red{\begin{prop}
%Optimization problem defined in~(\ref{eq:max-mcd}) under MCD has a unique global optimum.
%\end{prop}}

\subsection{Upper Bound (UB) on optimal aggregate utility}\label{sec:line-ub}
Constraint \eqref{eq:hpbmcdp2} in \eqref{eq:max-mcdp} is enforced due to variable exchanges under MCDP as discussed above. %in~(\ref{eq:mcdp-constraint}).  
 Here we can define an upper bound on optimal aggregate utility by removing \eqref{eq:hpbmcdp2} and solving the following optimization problem
\begin{align}\label{eq:max-ub}
%\text{\bf{L-U-MCDP:}} 
&\max \sum_{i\in \mathcal{D}} \sum_{l=1}^{|p|} \psi^{|p|-l} U_i(\lambda_ih_{il}),\nonumber\displaybreak[0]\\
&\text{s.t., constraints~(\ref{eq:hpbmcdp1}) and~(\ref{eq:hpbmcdp3})}.
%\text{s.t.} \quad&\sum_{i\in \mathcal{D}} h_{il} \leq B_l,\quad l=1, \cdots, |p|,  \displaybreak[1]\label{eq:hpbub1}\\
%&0\leq h_{il}\leq1, \quad\forall i \in \mathcal{D}, \quad l=1, \cdots, |p|.\label{eq:hpbub2}
\end{align}
Note that the UB optimization problem is now independent of any timer driven caching policy and can be used as a performance benchmark for comparing various caching policies.  Furthermore, it is easier to solve UB optimization problem~(\ref{eq:max-ub}) since it involves a smaller number of constraints compared to MCDP based optimization problem~(\ref{eq:max-mcdp}).

\subsection{Distributed Algorithm}\label{sec:line-online-primal}
In Section~\ref{sec:line-utility-max}, we formulated convex utility maximization problems with a fixed cache size.  However, system parameters (e.g. cache size and request processes) can change over time, so it is not feasible to solve the optimization offline and implement the optimal strategy.  Thus, we need to design {distributed} algorithms to implement the optimal strategy and adapt to the changes in the presence of limited information.  %In the following, we develop such an algorithm for MCDP.  %A similar algorithm exists for MCD and is omitted due to space constraints. 

\noindent{\textit{\textbf{Primal Algorithm:}}}
We aim to design an algorithm based on the optimization problem in~(\ref{eq:max-mcdp}), which is the primal formulation. Denote $\boldsymbol h_i=(h_{i1},\cdots, h_{i|p|})$ and $\boldsymbol h=(\boldsymbol h_1,\cdots, \boldsymbol h_n).$  We first define the following objective function.
\begin{align}\label{eq:primal1}
Z(\boldsymbol h) = &\sum_{i\in \mathcal{D}} \sum_{l=1}^{|p|} \psi^{|p|-l} U_i(\lambda_ih_{il})-\sum_{l=1}^{|p|}C_l\left(\sum\limits_{i\in \mathcal{D}} h_{il} - B_l\right)\nonumber\\
&-\sum_{i\in \mathcal{D}}\tilde{C}_i\left(\sum\limits_{l=1}^{|p|}h_{il}-1\right)-\sum_{i\in \mathcal{D}} \sum_{l=1}^{|p|}\hat{C}_{il}(-h_{il}),
\end{align}
where $C_l(\cdot), \tilde{C}_i(\cdot)$ and $\hat{C}_{il}(\cdot)$ are convex and non-decreasing penalty functions denoting the cost for violating constraints~(\ref{eq:hpbmcdp1}) and~(\ref{eq:hpbmcdp2}). 

{Note that constraint~(\ref{eq:hpbmcdp2}) ensures $h_{il} \le 1\; \forall i \in \mathcal{D},\;l=1, \cdots, |p|,$ provided $h_{il} \ge 0.$ One can assume that $h_{il} \ge 0$ holds in writing down \eqref{eq:primal1}. This would be true, for example, if the utility function is a $\beta$-fair utility function with $\beta > 0$ (Section $2.5$\cite{srikant13}). For other utility functions, it is challenging to incorporate constraint (\ref{eq:hpbmcdp3}) since it introduces $n|p|$ additional price functions. For all cases evaluated across various system parameters we found $h_{il} \ge 0$ to hold  true. Hence we ignore constraint (\ref{eq:hpbmcdp3}) in the primal formulation and define the following objective function}
\begin{align}\label{eq:primal}
Z(\boldsymbol h) = &\sum_{i\in \mathcal{D}} \sum_{l=1}^{|p|} \psi^{|p|-l} U_i(\lambda_ih_{il})-\sum_{l=1}^{|p|}C_l\left(\sum\limits_{i\in \mathcal{D}} h_{il} - B_l\right)\nonumber\\
&-\sum_{i\in \mathcal{D}}\tilde{C}_i\left(\sum\limits_{l=1}^{|p|}h_{il}-1\right).
\end{align}

%\begin{figure*}[htbp]
%\centering
%\begin{minipage}{.24\textwidth}
%\centering
%\includegraphics[width=\linewidth]{figures/n100_alpha08.pdf}
%   \caption{Hit probability for MCDP in a three-node path. }
%\label{mcdphp-line1}
%\end{minipage}\hfill
%\begin{minipage}{.24\textwidth}
%\centering
% \includegraphics[width=\linewidth]{figures/n100_primal_convergence.pdf}
%    \caption{Convergence of primal algorithm.}
%\label{mcdpcs-line2}
%\end{minipage}\hfill
%\begin{minipage}{.24\textwidth}
%\centering
%\includegraphics[width=\linewidth]{figures/n100_compare_policies_alpha1point4.pdf}
%\caption{Optimal aggregated utilities  in a three-node path.}
%\label{mcdpcs-line3}
%\end{minipage}\hfill
%\begin{minipage}{.24\textwidth}
%\centering
%\includegraphics[width=\linewidth]{figures/n100_ubVsMcdp.pdf}
%\caption{Aggregate utility under MCDP and UB across different network parameters.}
%\label{fig:line-comp-lcd}
%\end{minipage}
%\vspace{-0.1in}
%\end{figure*}

It is clear that $Z(\cdot)$ is strictly concave. Hence, a natural way to obtain the maximal value of~(\ref{eq:primal}) is to use the standard \emph{gradient ascent algorithm} to move the variable $h_{il}$ for $i\in\mathcal{D}$ and $l\in\{1,\cdots,|p|\}$ in the direction of gradient, %$\frac{\partial Z(\boldsymbol h)}{\partial h_{il}}=$
\begin{align}
\frac{\partial Z(\boldsymbol h)}{\partial h_{il}}&=\lambda_i\psi^{|p|-l}U_i^\prime(\lambda_ih_{il})-C_l^\prime\left(\sum\limits_{j\in \mathcal{D}} h_{jl} - B_l\right) \nonumber\displaybreak[0]\\
&\qquad\qquad-\tilde{C}_i^\prime\left(\sum\limits_{m=1}^{|p|}h_{im}-1\right),
\end{align}
where $U_i^\prime(\cdot),$ $C_l^\prime(\cdot)$, $\tilde{C}_i^\prime(\cdot)$ denote partial derivatives w.r.t. $h_{il}.$

Since $h_{il}$ indicates the probability that content $i$ is in cache $l$, $\sum_{j\in \mathcal{D}} h_{jl}$ is the expected number of contents currently in cache $l$, denoted by $B_{\text{curr},l}$. 

Therefore, the primal algorithm for MCDP is given by
\begin{subequations}\label{eq:primal-mcdp}
\begin{align}
&T_{il}[k] \leftarrow 
\begin{cases}
    \frac{1}{\lambda_i}\log \bigg(1 + \frac{h_{il}[k]}{1-\big(h_{i1}[k] + h_{i2}[k] + \cdots + h_{i|p|}[k]\big)}\bigg),\;\; l=1;\\
    \frac{1}{\lambda_i}\log \bigg(1 + \frac{h_{il}[k]}{h_{i(l-1)}[k]}\bigg),\quad l= 2, \cdots , |p|,\label{eq:primal-mcdp-t}
\end{cases}\displaybreak[0]\\
&h_{il}[k+1]\leftarrow \max\Bigg\{0, h_{il}[k]+\zeta_{il}\Bigg[\lambda_i\psi^{|p|-l}U_i^\prime(\lambda_ih_{il}[k])\nonumber\displaybreak[1]\\
&\qquad\qquad-C_l^\prime\left(B_{\text{curr},l} - B_l\right)-\tilde{C}_i^\prime\left(\sum\limits_{m=1}^{|p|}h_{im}[k]-1\right)\Bigg]\Bigg\},\label{eq:primal-mcdp-h}
\end{align} 
\end{subequations}
where $\zeta_{il} > 0$ is the step-size parameter, and $k$ is the iteration number incremented upon each request arrival. 
\begin{theorem}\label{thm:mcdp-primal-conv}
The primal algorithm~(\ref{eq:primal-mcdp}) is locally asymptotically stable given a sufficiently small step-size parameter $\zeta_{il}.$
\end{theorem}
\begin{proof}
Since $U_i(\cdot)$ is strictly concave, $C_l(\cdot)$ and $\tilde{C}_i(\cdot)$ are convex, then~(\ref{eq:primal}) is strictly concave, hence there exists a unique maximizer.  Denote it as $\boldsymbol h^*.$ 

Any differentiable function $f(x)$ can be linearized around a point $x^*$ as $L(x) = f(x^*) + f^{\prime}(x^*)(x-x^*).$ Denote $\forall i, l$
\begin{align}\label{eq:fun}
f(h_{il}) = h_{il} +\zeta_{il}\Bigg[&\lambda_i\psi^{|p|-l}U_i^\prime(\lambda_ih_{il})-C_l^\prime\left(\sum\limits_{j\in \mathcal{D}} h_{jl} - B_l\right)\nonumber\\&-\tilde{C}_i^\prime\left(\sum\limits_{m=1}^{|p|}h_{im}-1\right)\Bigg],
\end{align}
\noindent with $f:\mathbb{R}^+\rightarrow \mathbb{R}.$ We have $f(h_{il}^*) = h_{il}^*$. Under linearization,
\begin{align}
h_{il}[k+1] = h_{il}^* +  f^{\prime}(h_{il}^*)(h_{il}[k]- h_{il}^*). \label{eq:lin12}
\end{align}
\noindent Denote $h_{il}^{\delta}[k] = h_{il}[k]-h_{il}^*$ as deviation from $h_{il}^*$ at $k^{th}$ iteration. Hence we have
\begin{align}
h_{il}^{\delta}[k+1] =f^{\prime}(h_{il}^*)h_{il}^{\delta}[k]=\left[f^{\prime}(h_{il}^*)\right]^k h_{il}^{\delta}[0].\label{eq:lin2}
\end{align}

Thus \eqref{eq:lin12} is locally asymptotically stable if 
%Assuming $\eta_{\delta}^{(0)}\approx 0$, $\eta_{\delta}^{(k+1)}\approx \eta_{\delta}^{(k)}$ only when $|1-\frac{\gamma B^2}{W}| < 1$ with $\gamma,B,W >0.$ Thus for local asymptotic stability we have
\begin{align}
|f^{\prime}(h_{il}^*)| < 1\label{eq:gamma_lstb1}.
\end{align}
Computing $f^{\prime}(h_{il}^*)$ from \eqref{eq:fun} and substituting in \eqref{eq:gamma_lstb1} yields 
{\footnotesize
\begin{align}
\zeta_{il} < \frac{2}{C_l^{\prime\prime}\left(\sum\limits_{j\in \mathcal{D}} h_{jl}^* - B_l\right)+\tilde{C}_i^{\prime\prime}\left(\sum\limits_{m=1}^{|p|}h_{im}^*-1\right) - \lambda_i^2\psi^{|p|-l}U_i^{\prime\prime}(\lambda_ih_{il}^*)}\label{eq:gamma_lstb}.
\end{align}
}
Note that, since the functions $C_l, \tilde{C}_i$ and $U_i$ are strictly convex, strictly convex and strictly concave functions respectively, $C_l^{\prime\prime}(x) < 0, \tilde{C}_i^{\prime\prime}(x) < 0$ and $ U_i^{\prime\prime}(x) < 0 \;\forall x \in \mathbb{R}^+.$ Hence the r.h.s of \eqref{eq:gamma_lstb} is strictly positive and for a sufficiently small positive step-size parameter $\zeta_{il}$, \eqref{eq:gamma_lstb} always holds. Thus the update rule \eqref{eq:primal} converges to $\boldsymbol h^*$ as long as $h_{il}^{(0)}$ is sufficiently close to $h_{il}^*$ for all $i \in \mathcal{D}$ and $l = 1,2,\cdots,|p|.$ 
\end{proof}

\begin{remark}
Note that the primal formulation in \eqref{eq:primal-mcdp} can be implemented distributively with respect to (w.r.t.) different contents and caches by some amount of book-keeping. For example in \eqref{eq:primal-mcdp-h}, $\sum_{m=1}^{|p|}h_{im}[k]$ at cache $l$ can be computed by first storing the value of $\sum_{m=1}^{|p|}h_{im}[k]$ at the edge cache in previous iteration and updating it during delivery of content $i$ (from cache $l$) to the user.
\end{remark}

%Since $U_i(\cdot)$ is strictly concave, $C_l(\cdot)$ and $\tilde{C}_i(\cdot)$ are convex, ~(\ref{eq:primal}) is strictly concave, hence there exists a unique maximizer.  Denote it as $\boldsymbol h^*.$ 

%Define the following function
%\begin{align}
%Y(\boldsymbol h)=Z(\boldsymbol h^*)-Z(\boldsymbol h),
%\end{align} 
%then it is clear that $Y(\boldsymbol h)\geq 0$ for any feasible $\boldsymbol h$ that satisfies the constraints in the original optimization problem, and  $Y(\boldsymbol h)= 0$ if and only if $\boldsymbol h=\boldsymbol h^*.$
%
%We prove that $Y(\boldsymbol h)$ is a Lyapunov function, and then the above primal algorithm converges to the optimum.  Details are available in Appendix~\ref{appendix:convergence}. 

\begin{figure}
  \begin{subfigure}[b]{0.49\columnwidth}
    \includegraphics[width=\linewidth]{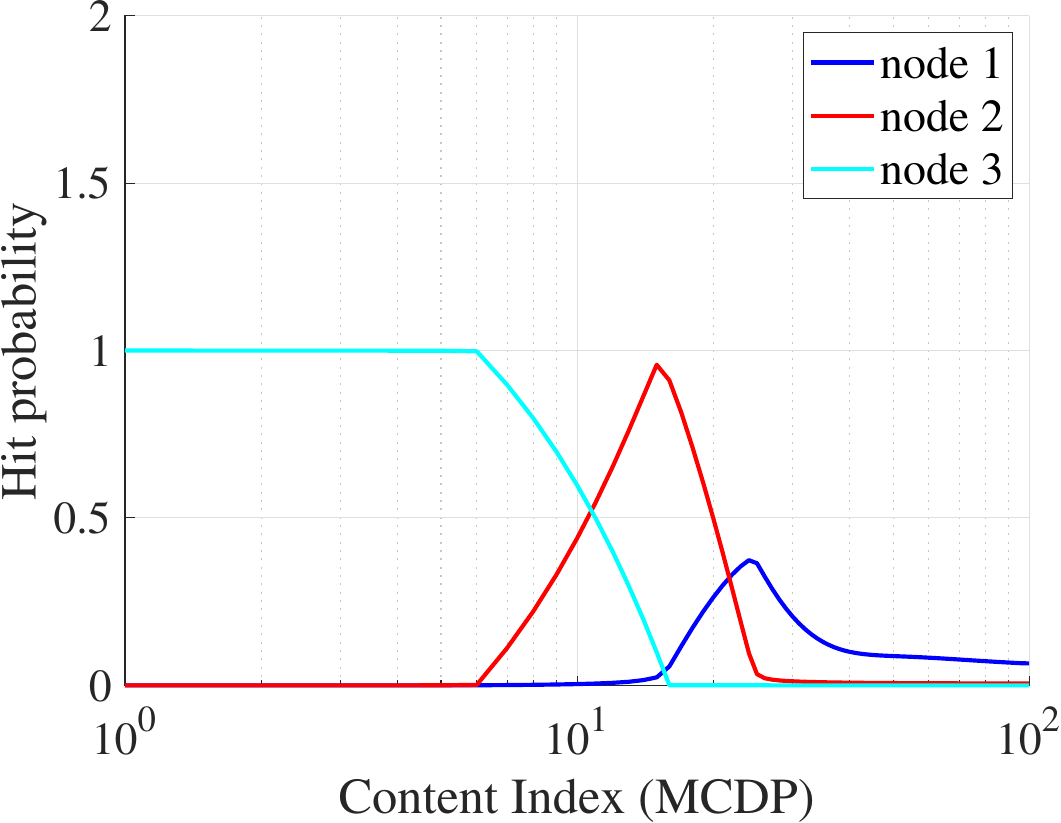}
   %\subcaption{Hit probability for MCDP in a three-node path. }
\label{mcdphp-line1}
  \end{subfigure}
  \hfill %%
  \begin{subfigure}[b]{0.49\columnwidth}
    \includegraphics[width=\linewidth]{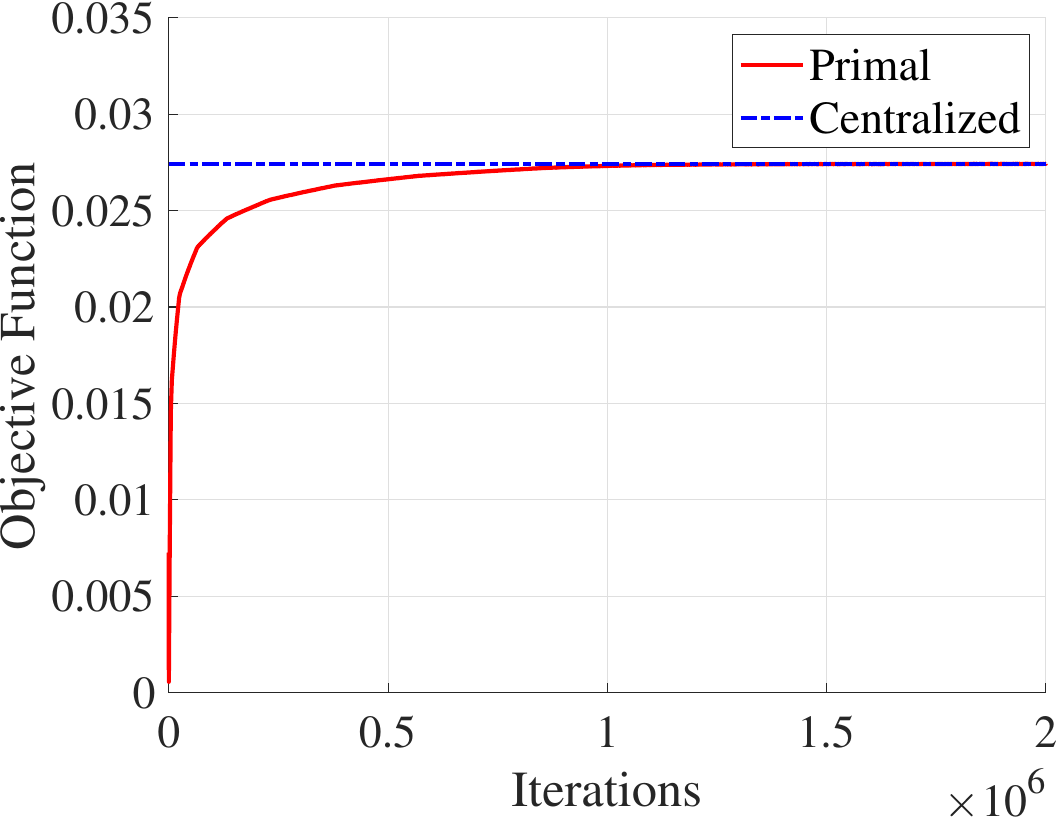}
    %\subcaption{Convergence of primal algorithm.}
\label{mcdpcs-line2}
  \end{subfigure}
  \caption{\textit{(Left)}: (a) Hit probability for MCDP in a three-node path; \textit{(Right)}: (b) Convergence of primal algorithm.}
  \label{mcdp-line}
  \vspace{-0.1in}
\end{figure}
\subsection{Model Validation and Insights}\label{sec:validations-line-cache}
We validate our analytical results with simulations for MCDP.  We consider a three-node path with cache capacities $B_l= 10$, $l=1, 2, 3.$  The total number of unique contents considered in the system is $n=100.$ We consider the Zipf popularity distribution with parameter $\alpha=0.8$.  W.l.o.g., we consider log based utility function\footnote{One can also choose $U_i(x) = \lambda_i\log x.$ However, $U_i(x)$ evaluated at $x=0$ becomes negative infinity, which may produce undesired results while comparing the performance of MCDP with other caching policies.} $U_i(x) = \lambda_i\log(1+x)$ \cite{vecer18}, and discount factor $\psi=0.1.$  We assume that requests arrive according to a Poisson process with aggregate request rate $\Lambda=1.$

We first solve optimization problem~(\ref{eq:max-mcdp}) using a Matlab routine \texttt{fmincon}.  From Figure~\ref{mcdp-line} (a), we observe that popular contents are assigned higher hit probabilities at cache node $3,$ i.e. at the edge cache closest to the user as compared to other caches. The optimal hit probabilities assigned to popular contents at other caches are almost negligible. However, the assignment is reversed for moderately popular contents. For non-popular contents, optimal hit probabilities at cache node $1$ (closest to origin server) are the highest.

We then implement our primal algorithm given in~(\ref{eq:primal-mcdp}), where we take the following penalty functions \cite{srikant13} $C_l(x)= \max\{0,x - B_l\log(B_l+x)\}$ and $\tilde{C}_i(x)= \max\{0,x - \log(1+x)\}$. Figure~\ref{mcdp-line} (b) shows that the primal algorithm successfully converges to the optimal solution. %\blue{convergence speed}.

\begin{figure}
\centering
\includegraphics[width=0.8\columnwidth]{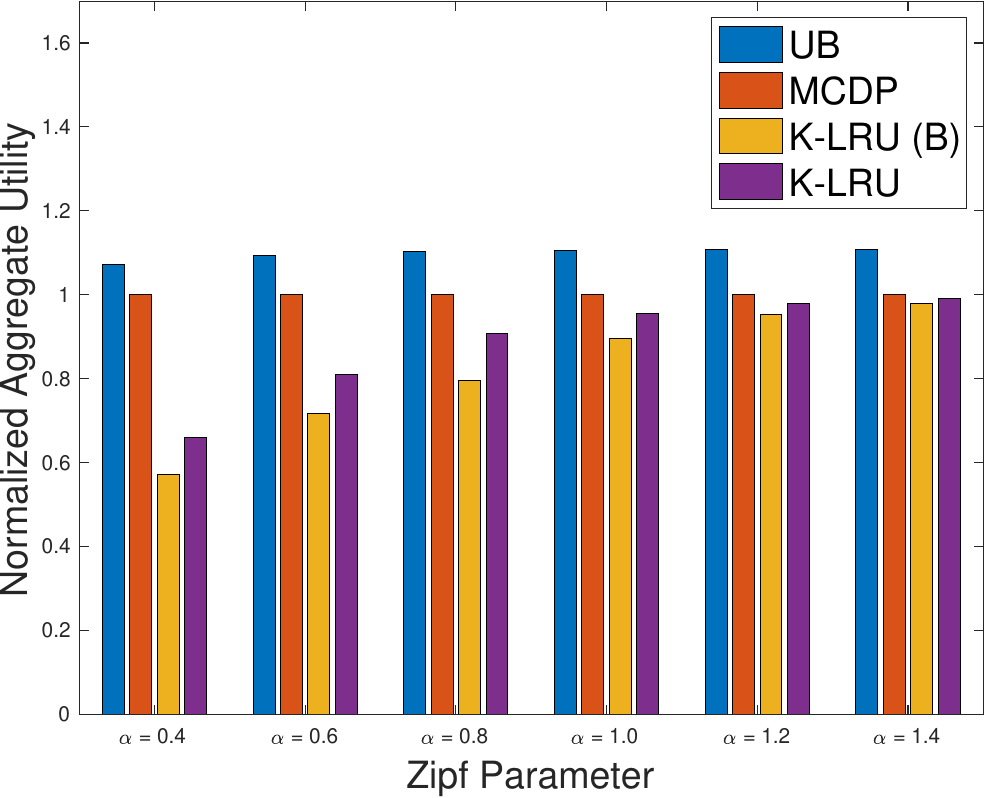}
\caption{Optimal aggregated utilities  in a three-node path.}
%\caption{Optimal aggregated utilities under different caching eviction policies, normalized to the aggregated utilities under MCDP.}
\label{mcdpcs-line3}
\vspace{-0.1in}
\end{figure}

\begin{figure}
\centering
\includegraphics[width=1\linewidth]{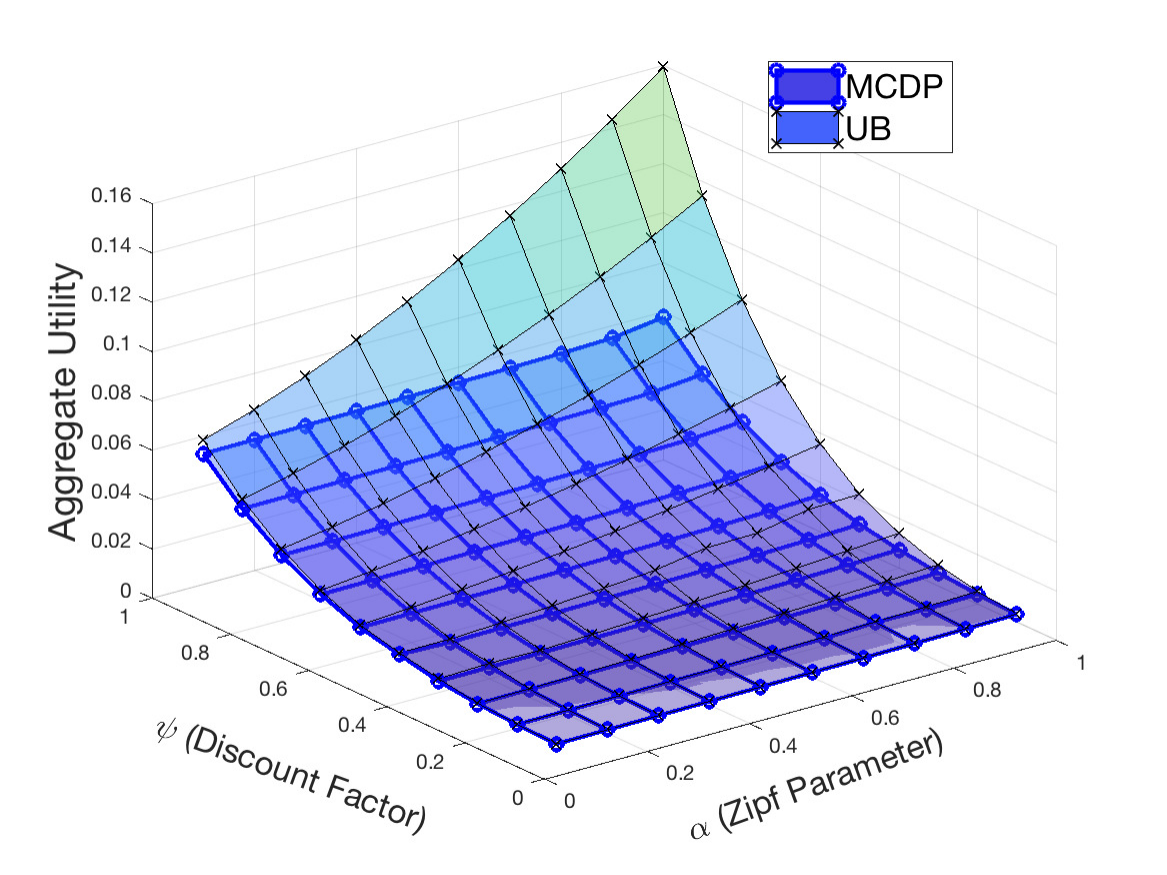}
\caption{Aggregate utility under MCDP and UB across different network parameters.}
\label{fig:line-comp-lcd}
\vspace{-0.2in}
\end{figure}

We also compare the performance of MCDP to other policies such as K-LRU (K=3), K-LRU with big cache abstraction: K-LRU(B) and the UB based bound.   We plot the relative performance w.r.t. the optimal aggregated utilities of all above policies, normalized to that under MCDP shown in Figure \ref{mcdpcs-line3}.   We observe that MCDP significantly outperforms K-LRU and K-LRU(B) for low and moderate values of Zipf parameter.  Furthermore, the performance gap between UB and MCDP increases in Zipf parameter.

Finally, we consider the effect of $\alpha$ and $\psi$ on the performance gap (difference between optimal aggregate utility) between UB and MCDP. We present the simulation results in  Figure \ref{fig:line-comp-lcd}.   Note that utility under both UB and MCDP increases in either $\alpha$ or $\psi$ or both.  When either $\alpha$ or $\psi$ is small, irrespective of the value of the other, the performance gap is minor or negligible.  However, the gap is considerably large when both $\alpha$ and $\psi$ are high. We believe, due to high communication overhead between successive layers of caches, $\psi$ has a small value, thus indicating a minor performance gap between MCDP and UB.

%\blue{Need add some discussions why we do not compare with LRU etc.} 
%
%\blue{Descriptions on Figure~\ref{fig:line-comp-lcd}?}
 
%\subsection{Minimizing Overall Costs}\label{sec:line-cost-min}
%\input{03D-min-cost}

\section{General Edge Networks}\label{sec:general}

In Section~\ref{sec:line}, we consider linear edge networks and characterize the optimal TTL policy for content when coupled with MCDP.   %We will use these results and insights to extend this to general edge networks in this section.
Inspired by these results, we consider general edge networks in this section. 

\subsection{Contents, Servers and Requests}
We consider the general edge network described in Section~\ref{sec:prelim}.  Denote by $\mathcal{P}$ the set of all requests, and $\mathcal{P}^i$ the set of requests for content $i.$   Suppose a cache in node $v$ serves two requests $(v_1, i_1, p_1)$ and $(v_2, i_2, p_2)$, then there are two cases: (i) non-common requested content, i.e., $i_1\neq i_2;$ and (ii) common requested content, i.e., $i_1=i_2.$

\subsection{Non-common Requested Content}\label{sec:non-common}
In this section, we consider the case that each network cache serves requests for different contents from each request $(v, i, p)$ passing through it.   %\red{Since there is no coupling between different requests $(v, i, p),$ we can directly generalize the results for linear cache networks in Section~\ref{sec:line-cache}. }  
%Denote the timer associated with each content $i$ as $T_{il}$ at each cache $l\in\{1,\cdots,|p|\}$ on the path.  Denote $\boldsymbol T_i=(T_{i1},\cdots, T_{i|p|})$ and $\boldsymbol T=(\boldsymbol T_1,\cdots, \boldsymbol T_n)$. 
%Without loss of generality ({\it W.l.o.g.}), we first focus on how to design optimal TTL policy on a particular path $p$ in Section~\ref{sec:line}, and then consider the edge cache network shown in Figure~\ref{fig:model} in Section~\ref{sec:non-common-general}. 
%\subsection{Optimal TTLs on Path $p$}\label{sec:line}
%\input{03A-line-cache}\label{sec:line-cache}
%\subsection{Optimal TTLs in Edge Networks}\label{sec:non-common-general}
Since there is no coupling between different requests $(v, i, p),$ we can directly generalize the results for a particular path $p$ in Section~\ref{sec:line} to a tree network.  Hence, given the utility maximization formulation in~(\ref{eq:max-mcdp}), we can directly formulate the optimization problem for MCDP as 
\begin{subequations}\label{eq:max-mcdp-general}
\begin{align}
%\text{\bf{G-N-U-MCDP:}} 
\max \quad&\sum_{i\in \mathcal{D}} \sum_{p\in\mathcal{P}^i}\sum_{l=1}^{|p|} \psi^{|p|-l} U_{ip}(\lambda_{ip}h_{il}^{(p)}) \displaybreak[0]\\
\text{s.t.} \quad&\sum_{i\in \mathcal{D}} \sum_{p:l\in\{1,\cdots,|p|\}}h_{il}^{(p)}\leq B_l,\quad\forall l\in V,   \displaybreak[1]\label{eq:hpbmcdp1-general}\\
& \sum_{l=1}^{|p|}h_{il}^{(p)}\leq 1,\quad\forall  i \in \mathcal{D}, p\in\mathcal{P}^i,\displaybreak[2]\label{eq:hpbmcdp2-general}\\
&0\leq h_{il}^{(p)}\leq1, \quad\forall i \in \mathcal{D}, l\in\{1,\cdots,|p|\},  p\in\mathcal{P}^i.
\end{align}
\end{subequations}
%where~(\ref{eq:hpbmcdp1-general}) is the cache capacity constraint and~(\ref{eq:hpbmcdp2-general}) follows the discussion for MCDP in~(\ref{eq:mcdp-constraint}). 

\begin{prop}
Since the feasible sets are convex and the objective function is strictly concave and continuous, the optimization problem defined in~(\ref{eq:max-mcdp-general}) under MCDP has a unique global optimum.
\end{prop}

%We can similarly formulate a utility maximization optimization problem for MCD for a general cache network.   This can be found in Appendix~\ref{appendix:mcd-non-common}.

%\blue{Add a primal algorithm for the general cache network with non-common requested content}

\subsubsection{Model Validation and Insights}
%\red{We consider a three-layer edge network shown in Figure~\ref{fig:cache-7-nodes-non} with node set $\{1,\cdots, 7\}$, which is  consistent with the YouTube video delivery system \cite{ramadan17,sasikumar19}. } 
{We consider a three-layer edge network shown in Figure~\ref{fig:model} with node set $\{1,\cdots, 7\}$, which is  consistent with the YouTube video delivery system \cite{ramadan17,sasikumar19}. Nodes $1$-$4$ are edge caches, and node $7$ is tertiary cache.} There exist four paths $p_1=\{1, 5, 7\},$ $p_2=\{2, 5, 7\},$ $p_3=\{3, 6, 7\}$ and  $p_4=\{4, 6, 7\}.$  Each edge cache serves requests for $100$ distinct contents, and cache size is $B_v=10$ for $v\in\{1,\cdots, 7\}.$  Assume that content follows a Zipf distribution with parameter $\alpha_1=0.2,$ $\alpha_2=0.4,$ $\alpha_3=0.6$ and $\alpha_4=0.8,$ respectively.  We consider utility function $U_{ip}(x) = \lambda_{ip}\log(1+x),$ where $\lambda_{ip}$ is the request arrival rate for content $i$ on path $p,$ and requests are described by a Poisson process with $\Lambda_p=1$ for $p=1, 2, 3, 4.$  The discount factor $\psi=0.1$.

Figure~\ref{7nodecache-non} shows results for path $p_4=\{4, 6, 7\}.$ From Figure~\ref{7nodecache-non} (Left), we observe that our algorithm yields the exact optimal and empirical hit probabilities under MCDP.  Figure~\ref{7nodecache-non} (Right) shows the probability density for the number of contents\footnote{The constraint~(\ref{eq:hpbmcdp1-general}) in problem~(\ref{eq:max-mcdp-general}) is on average cache occupancy.  However it can be shown that if $n\to\infty$ and $B_l$ grows in sub-linear manner, the probability of violating the target cache size $B_l$ becomes negligible \cite{dehghan16}.} in the edge network.  As expected, the density is concentrated around their corresponding cache sizes. Similar trends exist for paths $p_1$, $p_2$ and $p_3$, hence are omitted here.

\begin{figure}
  \begin{subfigure}[b]{0.49\columnwidth}
    \includegraphics[width=\linewidth]{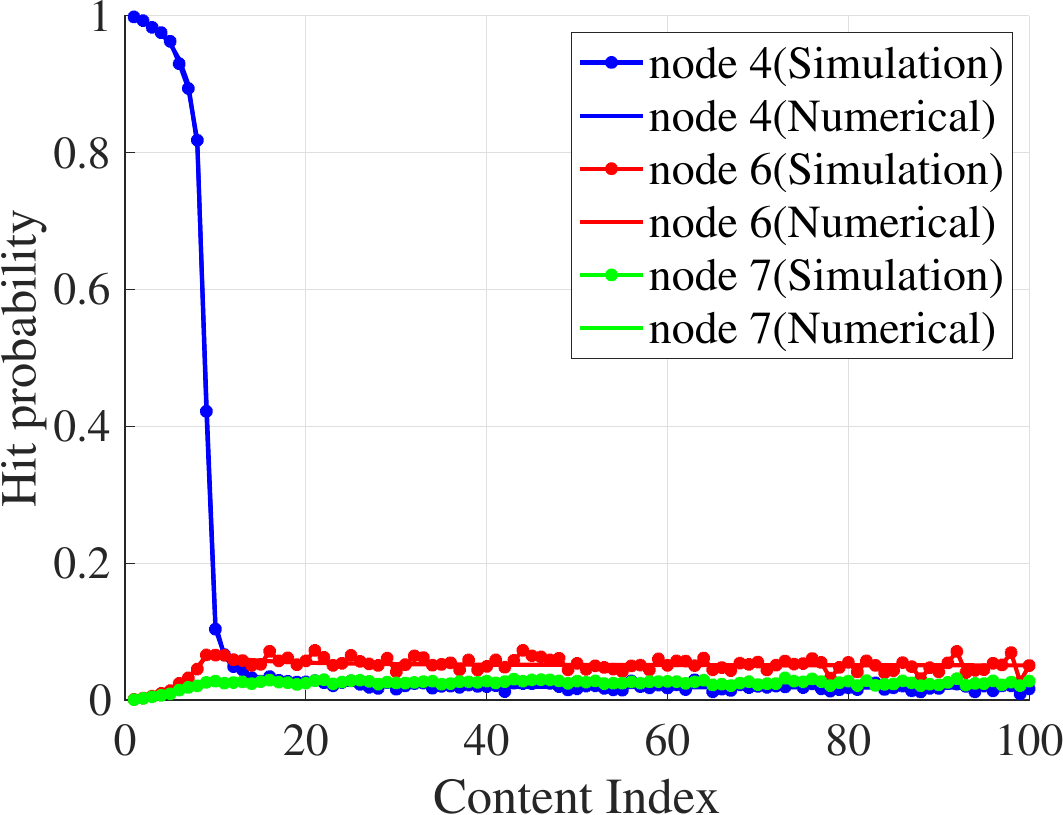}
   %\subcaption{Hit probability of MCDP under three-layer edge network where each path requests distinct contents.}
\label{7nodecache-hit-non}
  \end{subfigure}
  \hfill %%
  \begin{subfigure}[b]{0.49\columnwidth}
    \includegraphics[width=\linewidth]{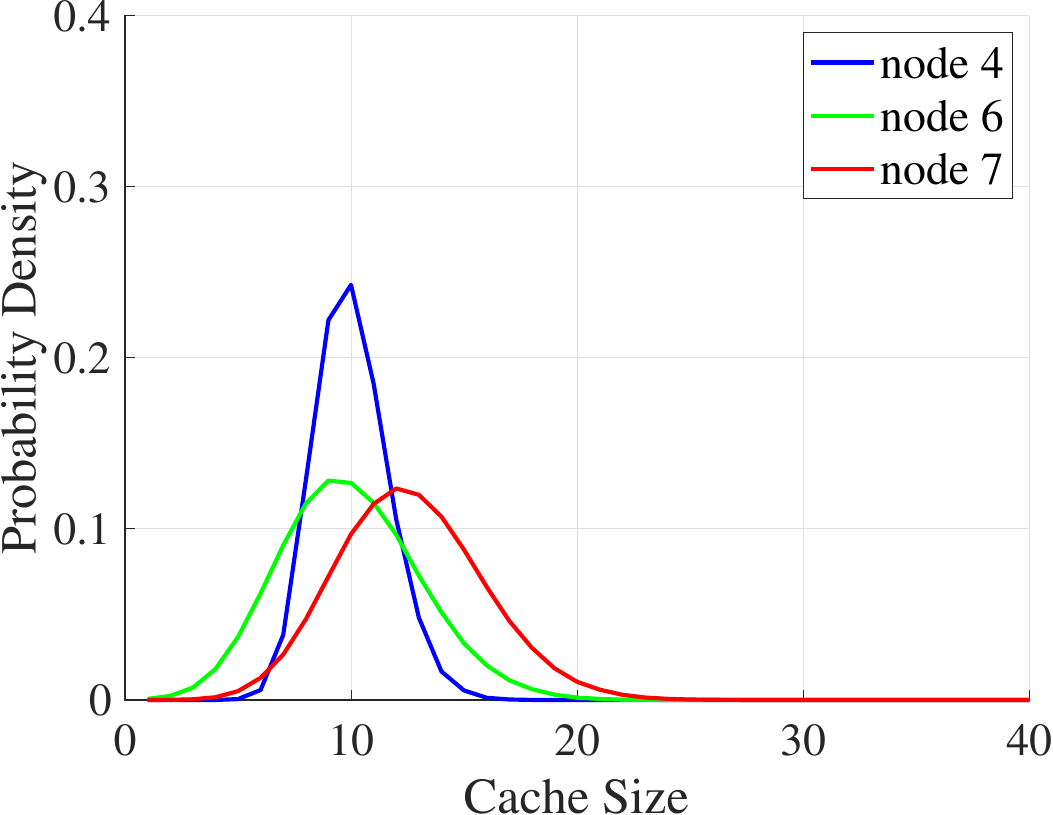}
    %\subcaption{Cache size of MCDP under three-layer edge  network where each path requests distinct contents.}
\label{7nodecache-size-non}
  \end{subfigure}
  \vspace{-0.1in}
  \caption{\textit{(Left)}: (a) Hit probability; \textit{(Right)}: (b) Cache size, of MCDP under three-layer edge  network where each path requests distinct contents.}
  \label{7nodecache-non}
  \vspace{-0.1in}
\end{figure}

\subsection{Common Requested Contents}\label{sec:common}

%\subsection{Common Requested Contents}\label{sec:general-cache-common}
Now consider the case where different users share the same content, e.g., there are two requests $(v_1, i, p_1)$ and $(v_2, i, p_2).$  Suppose that cache $l$ is on both paths $p_1$ and $p_2$, where $v_1$ and $v_2$ request the same content $i$.   If we cache separate copies on each path,  results from the previous section apply.  However, maintaining redundant copies in the same cache decreases efficiency.  A simple way to deal with that is to only cache one copy of content $i$ at $l$ to serve both requests from $v_1$ and $v_2.$ Though this reduces redundancy, it complicates the optimization problem.

In the following, we formulate a utility maximization problem for MCDP with TTL caches, where all users share the same requested contents $\mathcal{D}.$ 
\begin{subequations}\label{eq:max-mcdp-general-common}
\begin{align}
%&\text{\bf{G-U-MCDP:}}\nonumber\displaybreak[0]\\
\max \quad&\sum_{i\in \mathcal{D}} \sum_{p\in \mathcal{P}^i} \sum_{l=1}^{|p|} \psi^{|p|-l} U_{ip}(\lambda_{ip}h_{il}^{(p)}) \displaybreak[0]\\
\text{s.t.} \quad& \sum_{i\in \mathcal{D}} \bigg(1-\prod_{p:j\in\{1,\cdots,|p|\}}(1-h_{ij}^{(p)})\bigg) \leq B_j,\quad\forall j \in V, \displaybreak[1]\label{eq:max-mcdp-genenral-cons1}\\
& \sum_{j\in\{1,\cdots,|p|\}}h_{ij}^{(p)}\leq 1,\quad \forall i \in \mathcal{D},  p \in \mathcal{P}^i, \displaybreak[2] \label{eq:max-mcdp-genenral-cons2}\\
&0\leq h_{il}^{(p)}\leq 1, \quad\forall i\in\mathcal{D},  j\in\{1,\cdots,|p|\},  p\in\mathcal{P}^i,\label{eq:max-mcdp-genenral-cons3}
\end{align}
\end{subequations}
where~(\ref{eq:max-mcdp-genenral-cons1}) ensures that only one copy of content $i\in\mathcal{D}$ is cached at node $j$ for all paths $p$ that pass through node $j$. This is because the term $1-\prod_{p:j\in\{1,\cdots,|p|\}}(1-h_{ij}^{(p)})$ is the overall hit probability of content $i$ at node $j$ over all paths. ~(\ref{eq:max-mcdp-genenral-cons2}) is the cache capacity constraint and~(\ref{eq:max-mcdp-genenral-cons3}) is the constraint from MCDP TTL cache policy as discussed in Section~\ref{sec:ttl-hit-prob}.

\begin{example}\label{exm}
Consider two requests $(v_1, i, p_1)$ and $(v_2, i, p_2)$ with paths $p_1$ and $p_2$ intersecting at $j.$  Let the corresponding path perspective hit probability be $h_{ij}^{(p_1)}$ and $h_{ij}^{(p_2)}$.  Then the term inside outer summation of~(\ref{eq:max-mcdp-genenral-cons1}) is $1-(1-h_{ij}^{(p_1)})(1-h_{ij}^{(p_2)})$, i.e., the hit probability of content $i$ in $j$. 
\end{example}

\begin{remark}
Note that we assume independence between different requests $(v, i, p)$ in~(\ref{eq:max-mcdp-general-common}), e.g.,  in Example~\ref{exm}, if the insertion of content $i$ in node $j$ is caused by request $(v_1, i, p_1),$ when request $(v_2, i, p_2)$ comes, it is not counted as a cache hit from its perspective.   Our framework still holds if we follow the logical TTL MCDP on a path.  However, in that case, the utilities will be larger than the one we consider here. 
\end{remark}

%Similarly, we can formulate a utility maximization optimization problem for MCD.  This can be found in Appendix~\ref{appendix:mcd-common}.

\begin{prop}
Since the feasible sets are non-convex, ~(\ref{eq:max-mcdp-general-common}) under MCDP is a non-convex optimization problem.
\end{prop}

In the following, we develop an optimization framework that handles the non-convexity issue in this optimization problem and provides a distributed solution.  To this end, we first introduce the Lagrangian function
\begin{align}\label{eq:lagrangian}
L(\boldsymbol{h,\nu,\mu})&=\sum_{i \in \mathcal{D}}\sum_{p \in \mathcal{P}^i}\sum_{l=1}^{|p|}\psi^{|p|-l}U_{ip}(\lambda_{ip}h_{il}^{(p)})\nonumber\displaybreak[0]\\
&-\sum_{j\in V}\nu_{j}\Bigg(\sum_{i \in \mathcal{D}}\bigg[1-\prod_{p:j\in\{1,\cdots,|p|\}}(1-h_{ij}^{(p)})\bigg]-B_j\Bigg)\nonumber\displaybreak[1]\\
&-\sum_{i\in\mathcal{D}}\sum_{p\in\mathcal{P}^i}\mu_{ip}\Bigg(\sum_{j\in\{1,\cdots,|p|\}}h_{ij}^{(p)}-1\Bigg),
\end{align}
where the Lagrangian multipliers (price vector and price matrix) are $\boldsymbol\nu=(\nu_{j})_{j\in V},$ and $\boldsymbol \mu=(\mu_{ip})_{i\in\mathcal{D}, p\in\mathcal{P}}.$  Constraint \eqref{eq:max-mcdp-genenral-cons3} is ignored in the Lagrangian function due to the same reason stated for the primal formulation in Section \ref{sec:line-online-primal}.

The dual function can be defined as 
%\begin{small}
\begin{align}\label{eq:dual}
d(\boldsymbol{\nu,\mu})=\sup_{\boldsymbol h} L(\boldsymbol{h,\nu,\mu}),
\end{align}
%\end{small}
and the dual problem is given as
\begin{align}\label{eq:dual-opt}
\min_{\boldsymbol{ \nu,\mu}} \quad&d(\boldsymbol{\nu,\mu})=L(\boldsymbol h^*(\boldsymbol{\nu,\mu}), \boldsymbol{\nu,\mu}),\quad\text{s.t.}\quad\boldsymbol{\nu,\mu}\geq \boldsymbol{0},
\end{align}
where the constraint is defined pointwise for $\boldsymbol{\nu,\mu},$ and $\boldsymbol h^*(\boldsymbol{\nu,\mu})$ is a function that maximizes the Lagrangian function for given $(\boldsymbol{\nu,\mu}),$ i.e.,
\begin{align}\label{eq:dual-opt-h}
\boldsymbol h^*(\boldsymbol{\nu,\mu})=\arg\max_{\boldsymbol h}L(\boldsymbol{h,\nu,\mu}).
\end{align}

The dual function $d(\boldsymbol{\nu,\mu})$ is always convex in $(\boldsymbol{\nu,\mu})$ regardless of the convexity of the optimization problem~(\ref{eq:max-mcdp-general-common}) \cite{boyd04}.  Therefore, it is always possible to iteratively solve the dual problem using
\begin{align}\label{eq:dual-opt-lambda}
&\nu_l[k+1]=\nu_l[k]-\gamma_l\frac{\partial L(\boldsymbol{\nu,\mu})}{\partial \nu_l},\nonumber\displaybreak[0]\\
&\mu_{ip}[k+1]= \mu_{ip}[k]-\eta_{ip}\frac{\partial L(\boldsymbol{\nu,\mu})}{\partial \mu_{ip}},
\end{align}
where $\gamma_l$ and $\eta_{ip}$ are the step sizes, and $\frac{\partial L(\boldsymbol{\nu,\mu})}{\partial \nu_l}$ and $\frac{\partial L(\boldsymbol{\nu,\mu})}{\partial \mu_{ip}}$ are the partial derivative of $L(\boldsymbol{\nu,\mu})$ w.r.t. $ \nu_l$ and $\mu_{ip},$ respectively, satisfying
\begin{align}\label{eq:gradient-lambda-mu}
\frac{\partial L(\boldsymbol{\nu,\mu})}{\partial \nu_l}
&=-\bigg(\sum_{i \in \mathcal{D}}\bigg[1-\prod_{p:l\in\{1,\cdots,|p|\}}(1-h_{il}^{(p)})\bigg]-B_l\bigg),\nonumber\displaybreak[0]\\
\frac{\partial L(\boldsymbol{\nu,\mu})}{\partial \mu_{ip}} &= -\Bigg(\sum_{j\in\{1,\cdots,|p|\}}h_{ij}^{(p)}-1\Bigg).
\end{align}

Sufficient and necessary conditions for the uniqueness of $\boldsymbol{\nu,\mu}$ are given in \cite{kyparisis85}.  The convergence of primal-dual algorithm consisting of~(\ref{eq:dual-opt-h}) and~(\ref{eq:dual-opt-lambda}) is guaranteed if the original optimization problem is convex. However, our problem is not convex. Nevertheless, we next show that the duality gap is zero, hence~(\ref{eq:dual-opt-h}) and~(\ref{eq:dual-opt-lambda}) converge to the globally optimal solution.  To begin with, we introduce the following results 

\begin{theorem}\label{thm:sufficient}
\cite{tychogiorgos13} (Sufficient Condition). If the price based function $\boldsymbol h^*(\boldsymbol{\nu,\mu})$ is continuous at one or more of the optimal lagrange multiplier vectors $\boldsymbol \nu^*$ and $\boldsymbol \mu^*$, then the iterative algorithm consisting of~(\ref{eq:dual-opt-h}) and~(\ref{eq:dual-opt-lambda}) converges to the globally optimal solution. 
\end{theorem}

\begin{theorem}\label{thm:necessary}
\cite{tychogiorgos13} If at least one constraint of~(\ref{eq:max-mcdp-general-common}) is active at the optimal solution, the condition in Theorem~\ref{thm:sufficient} is also a necessary condition. 
\end{theorem}

Hence, if we can show the continuity of $\boldsymbol h^*(\boldsymbol{\nu,\mu})$ and that constraints~(\ref{eq:max-mcdp-general-common}) are active, then given Theorems~\ref{thm:sufficient} and~\ref{thm:necessary}, the duality gap is zero, i.e., ~(\ref{eq:dual-opt-h}) and~(\ref{eq:dual-opt-lambda}) converge to the globally optimal solution.

Take the derivative of $L(\boldsymbol{h,\nu,\mu})$ w.r.t. $h_{il}^{(p)}$ for $i\in\mathcal{D},$ $l\in\{1,\cdots,|p|\}$ and $p\in\mathcal{P}^i$, we have %$\frac{\partial L(\boldsymbol{h,\nu,\mu})}{\partial h_{il}^{(p)}}=$
\begin{align}\label{eq:gradient}
\frac{\partial L(\boldsymbol{h,\nu,\mu})}{\partial h_{il}^{(p)}}&=\psi^{|p|-l}\lambda_{ip}U'_{ip}(\lambda_{ip}h_{il}^{(p)})-\mu_{ip}\nonumber\displaybreak[0]\\
&\qquad-\nu_{l}\Bigg(\prod_{\substack{q: q\neq p,\\ j\in\{1,\cdots, |q|\}}}(1-h_{ij}^{(q)})\Bigg).
\end{align}
Setting~(\ref{eq:gradient}) equal to zero, we obtain %$U'_{ip}(\lambda_{ip}h_{il}^{(p)})=$
\begin{align}\label{eq:gradient1}
U'_{ip}(\lambda_{ip}h_{il}^{(p)})=\frac{1}{\psi^{|p|-l}\lambda_{ip}}\Bigg(\nu_{l}\Bigg(\prod_{\substack{q: q\neq p,\\ j\in\{1,\cdots, |q|\}}}(1-h_{ij}^{(q)})\Bigg)+\mu_{ip}\Bigg).
\end{align}

%\begin{figure*}[htbp]
%\centering
%\begin{minipage}{.32\textwidth}
%\centering
%\includegraphics[width=1\columnwidth]{figures/primal_dual_n100_convergence.pdf}
%\caption{Convergence of Primal-Dual algorithm.}
%\label{fig:conv-primal-dual}
%\end{minipage}\hfill
%\begin{minipage}{.32\textwidth}
%\centering
%\includegraphics[width=1\textwidth]{figures/n100_compare_policies_tree.eps}
%\caption{Optimal aggregated utilities under common requested contents.}
%%\caption{Optimal aggregated utilities under different caching eviction policies, normalized to the aggregated utilities under MCDP.}
%\label{fig:policy-comp-tree}
%\end{minipage}\hfill
%\begin{minipage}{.32\textwidth}
%\centering
%\includegraphics[width=1\textwidth]{figures/tree_common_capVsutil.pdf}
%\caption{$\eta (= B_l/n)$ vs. overall objective under a three-layer edge network.}
%\label{fig:capVsutil}
%\end{minipage}
%\vspace{-0.2in}
%\end{figure*}

\begin{figure*}[htbp]
\centering
\begin{minipage}{.32\textwidth}
\centering
\includegraphics[width=\linewidth]{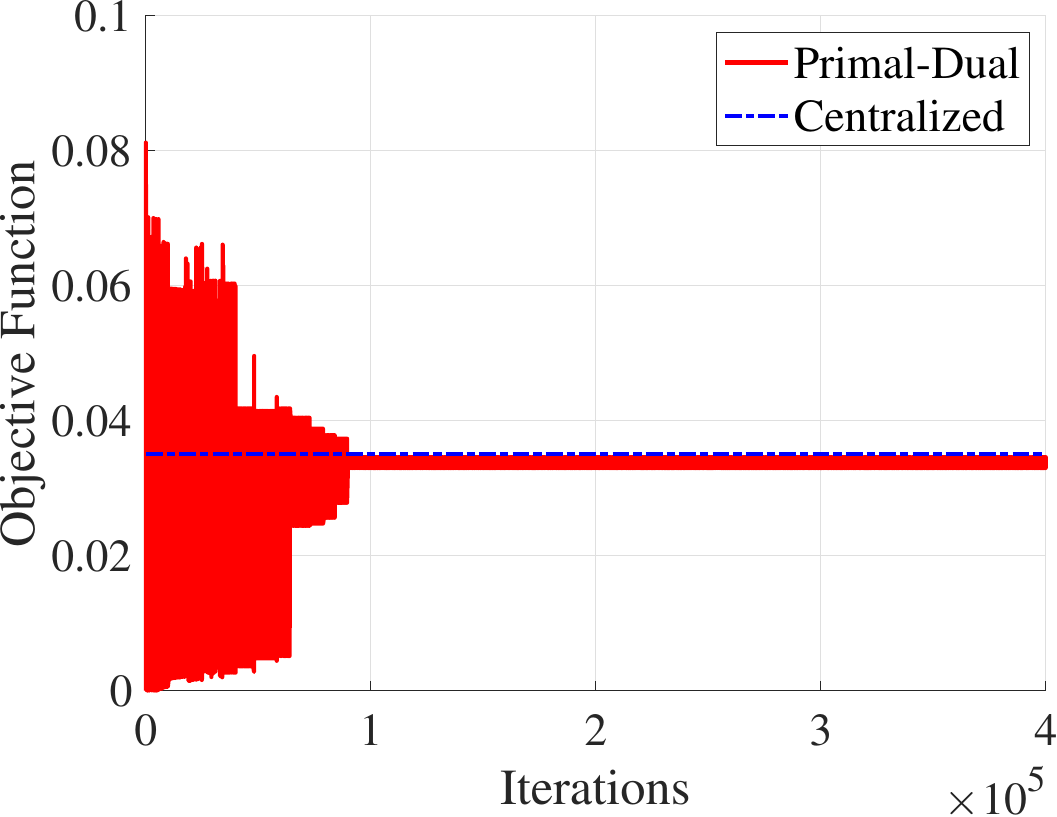}
\caption{Convergence of Primal-Dual algorithm.}
\label{fig:conv-primal-dual}
\end{minipage}\hfill
\begin{minipage}{.32\textwidth}
\centering
\includegraphics[width=\linewidth]{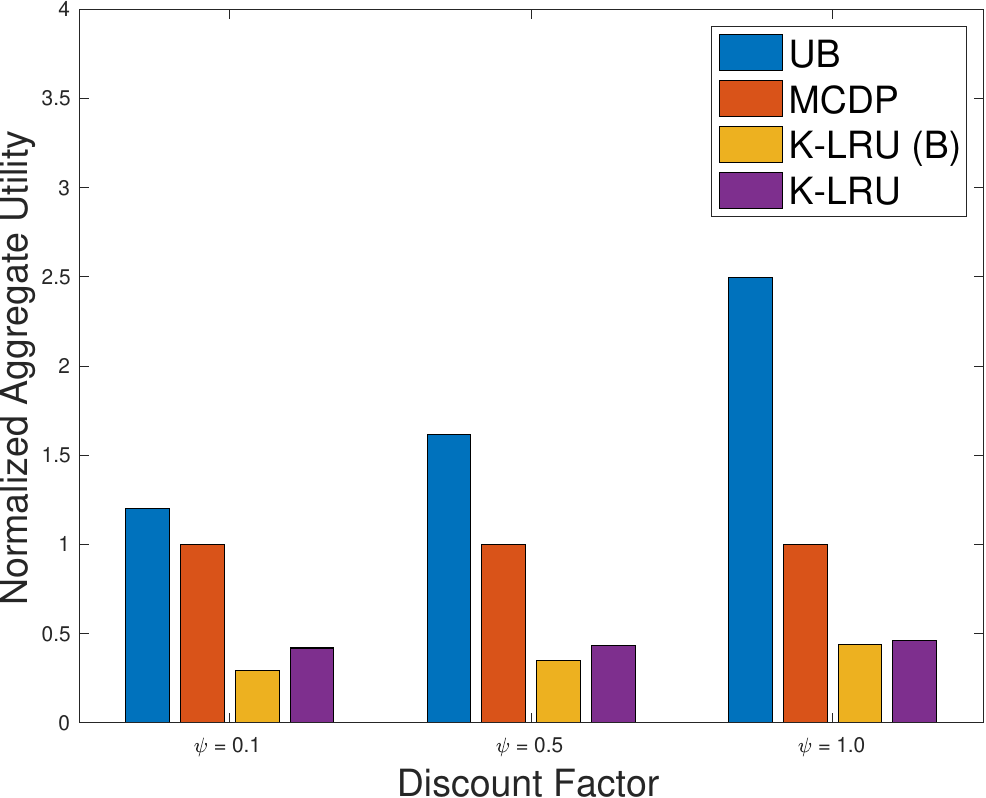}
\caption{Optimal aggregated utilities under common requested contents.}
\label{fig:policy-comp-tree}
\end{minipage}\hfill
\begin{minipage}{.32\textwidth}
\centering
\includegraphics[width=\linewidth]{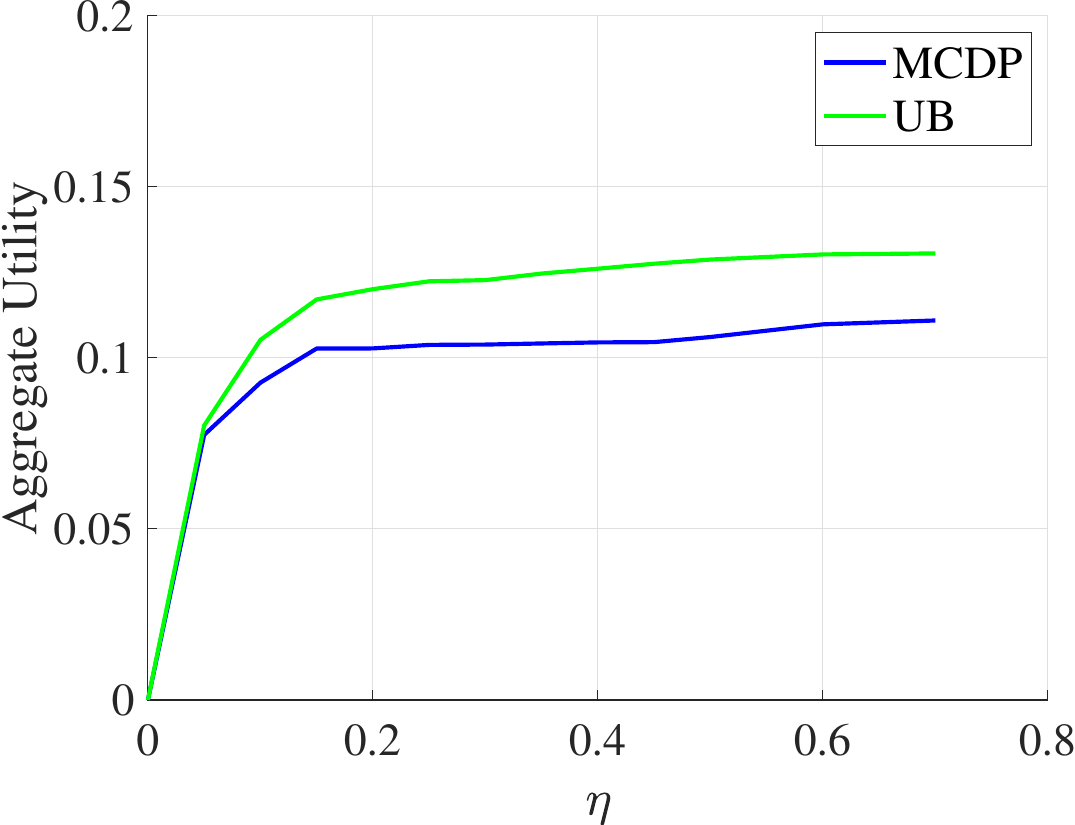}
\caption{$\eta (= B_l/n)$ vs. overall objective under a three-layer edge network.}
\label{fig:capVsutil}
\end{minipage}\hfill
%\begin{minipage}{.24\textwidth}
%\centering
%\includegraphics[height = 1.5in, width=1.7in]{figures/line_cost.png}
%\vspace{-0.1in}
%\caption{Various costs associated with MCDP policy.}
%\label{fig:model-linecost}
%\end{minipage}
\vspace{-0.1in}
\end{figure*}

Consider the utility function $U_{ip}(\lambda_{ip}h_{il}^{(p)})=w_{ip} \log(1+\lambda_{ip}h_{il}^{(p)})$, then $U'_{ip}(\lambda_{ip}h_{il}^{(p)})=w_{ip}/(1+\lambda_{ip}h_{il}^{(p)})$. Hence, from~(\ref{eq:gradient1}), we have
\begin{align}\label{eq:gradient2}
h_{il}^{(p)}=\frac{w_{ip}\psi^{|p|-l}}{\nu_{l}\Bigg(\prod_{\substack{q: q \neq p,\\ j\in\{1,\cdots, |q|\}}}(1-h_{ij}^{(q)})\Bigg)+\mu_{ip}} - \frac{1}{\lambda_{ip}}.
\end{align}

\begin{lemma}\label{lem:active}
Constraints~(\ref{eq:max-mcdp-genenral-cons1}) and~(\ref{eq:max-mcdp-genenral-cons2}) cannot be both non-active, i.e., at least one of them is active. 
\end{lemma}
\begin{proof}
We prove this lemma by contradiction. Suppose both constraints~(\ref{eq:max-mcdp-genenral-cons1}) and~(\ref{eq:max-mcdp-genenral-cons2}) are non-active, i.e., $\boldsymbol\nu=(\boldsymbol 0),$ and $\boldsymbol \mu=(\boldsymbol 0).$ Then the optimization problem~(\ref{eq:max-mcdp-general}) achieves its maximum when $h_{il}^{(p)}=1$ for all $i\in\mathcal{D},$ $l\in\{1,\cdots,|p|\}$ and $p\in\mathcal{P}^i.$  If so, then the left hand size of~(\ref{eq:max-mcdp-genenral-cons1}) equals $|\mathcal{D}|$ which is much greater than $B_l$ for $l\in V,$ which is a contradiction. Hence, constraints~(\ref{eq:max-mcdp-genenral-cons1}) and~(\ref{eq:max-mcdp-genenral-cons2}) cannot be both non-active.
\end{proof}  
  
From Lemma~\ref{lem:active}, we know that the feasible region for the Lagrangian multipliers satisfies $\mathcal{R}=\{\nu_{l}\geq0,\mu_{ip}\geq0,\nu_{l}+\mu_{ip}\neq0, \forall i\in\mathcal{D},l\in\{1,\cdots,|p|\},p\in\mathcal{P}^i\}.$

\begin{theorem}  
The hit probability $h_{il}^{(p)}$ given in~(\ref{eq:gradient2}) is continuous in $\nu_{l}$ and $\mu_{ip}$ for all $i\in\mathcal{D},$ $l\in\{1,\cdots,|p|\}$ and $p\in\mathcal{P}^i$ in the feasible region $\mathcal{R}.$
\end{theorem}
\begin{proof}
From Lemma~\ref{lem:active}, we know at least one of  $\nu_{l}$ and $\mu_{ip}$ is non-zero, for all $i\in\mathcal{D},$ $l\in\{1,\cdots,|p|\}$ and $p\in\mathcal{P}^i.$ Hence there are three cases, (i) $\nu_{l}\neq0$ and $\mu_{ip}=0$; (ii) $\nu_{l}=0$ and $\mu_{ip}\neq0$; and (iii) $\nu_{l}\neq0$ and $\mu_{ip}\neq0$.

For case (i), we have 
\begin{align}
h_{il}^{(p)}=\frac{w_{ip}\psi^{|p|-l}}{\nu_{l}\Bigg(\prod_{\substack{q: q \neq p,\\ j\in\{1,\cdots, |q|\}}}(1-h_{ij}^{(q)})\Bigg)}- \frac{1}{\lambda_{ip}},
\end{align}
which is clearly continuous in $\nu_{l},$ for all $i\in\mathcal{D},$ $l\in\{1,\cdots,|p|\}$ and $p\in\mathcal{P}^i.$

Similarly for case (ii), we have
%\begin{small} 
\begin{align}
h_{il}^{(p)}=\frac{w_{ip}\psi^{|p|-l}}{\mu_{ip}}- \frac{1}{\lambda_{ip}},
\end{align}
%\end{small}
which is also clearly continuous in $\mu_{ip},$ for all $i\in\mathcal{D},$ $l\in\{1,\cdots,|p|\}$ and $p\in\mathcal{P}^i$.

For case (iii), from~(\ref{eq:gradient2}), it is obvious that $h_{il}^{(p)}$ is continuous in $\nu_{l}$ and $\mu_{ip}$ for all $i\in\mathcal{D},$ $l\in\{1,\cdots,|p|\}$ and $p\in\mathcal{P}^i.$

Therefore, we know that $h_{il}^{(p)}$ is is continuous in $\nu_{l}$ and $\mu_{ip}$ for all $i\in\mathcal{D},$ $l\in\{1,\cdots,|p|\}$ and $p\in\mathcal{P}^i.$
\end{proof}

\begin{remark}
Note that similar arguments (by using Lemma \ref{lem:active}) hold true for various other choice of utility functions such as: $\beta$- fair utility functions (Section $2$\cite{srikant13}). Therefore, the primal-dual algorithm consisting of~(\ref{eq:dual-opt-h}) and~(\ref{eq:dual-opt-lambda}) converges to the globally optimal solution for a wide range of utility functions. 
\end{remark}

Algorithm~\ref{algo:primal-dual-alg} summarizes the details of this algorithm. 

\begin{algorithm}
	\begin{algorithmic}[]
		\State \textbf{Input}: $\forall \boldsymbol{\nu_0$, $\mu_0}\in\mathcal{R}$ and $\boldsymbol h_0$ 
		\State \textbf{Output}: The optimal hit probabilities $\boldsymbol h$
		\State \textbf{Step $0$:}  $t=0$, $\boldsymbol \nu[t]\leftarrow\boldsymbol \nu_0$, $\boldsymbol\mu[t]\leftarrow\boldsymbol\mu_0$, $\boldsymbol h[t]\leftarrow\boldsymbol h_0$
		\State \textbf{Step $t\geq 1$} 
	                 \While{Equation \eqref{eq:gradient-lambda-mu} $\neq$ 0}
		              \State \textbf{First,} compute $\boldsymbol h_{il}^{(p)}[t+1]$ for $i\in\mathcal{D},$ $l\in\{1,\cdots, |p|\}$ and $p\in\mathcal{P}^i$ through~(\ref{eq:gradient2});
		              \State \textbf{Second,} update $\nu_l[t+1]$ and $\mu_{ip}[t+1]$ through~(\ref{eq:dual-opt-lambda}) given $\boldsymbol h[t+1],$ $\boldsymbol\nu[t]$ and $\boldsymbol \mu[t]$ for $l\in V,$ $i\in\mathcal{D}$ and $p\in\mathcal{P}^i$
		\EndWhile
	\end{algorithmic}
	\caption{Primal-Dual Algorithm}
	\label{algo:primal-dual-alg}
\end{algorithm}

%\begin{figure}
%\centering
%\includegraphics[width=0.9\linewidth]{figures/opt_prob_dist_dual_vs_primal_new_new_n100.pdf}
%\caption{Hit probability for MCDP under a 7-node cache network, where each path shares common requested contents.}
%\label{fig:primal-dual}
%\vspace{-0.1in}
%\end{figure}

\subsubsection{Model Validation and Insights}\label{sec:general-common-validations}
We evaluate the performance of Algorithm~\ref{algo:primal-dual-alg} on a three-layer edge network shown in Figure~\ref{fig:model}.  We assume that there are totally $100$ unique contents in the system requested from four paths. The cache size is given as $B_v=10$ for $v=1,\cdots,7.$  We consider the utility function $U_{ip}(x) = \lambda_{ip}\log(1+x),$ and the popularity distribution over these contents is Zipf with parameter $0.8.$  W.l.o.g., the aggregate request arrival rate is one. The discount factor $\psi=0.1.$

%\begin{figure}
%\centering
%\includegraphics[width=0.9\linewidth]{figures/tree-comparison-LCD-new.pdf}
%\caption{Normalized optimal aggregated utilities under different caching eviction policies with LCD replications to that under MCDP.}
%\label{fig:tree-comp-lcd}
%\vspace{-0.1in}
%\end{figure}

We solve the optimization problem in~(\ref{eq:max-mcdp-general-common}) using a Matlab routine \texttt{fmincon}.  Then we implement our primal-dual algorithm given in Algorithm~\ref{algo:primal-dual-alg}.  {The result for aggregate optimal utility is presented in Figure~\ref{fig:conv-primal-dual}.  It is clear that the primal-dual algorithm successfully converges to the optimal solution.}

Similar to Section~\ref{sec:validations-line-cache}, we compare MCDP to K-LRU, K-LRU(B) and UB.   Figure~\ref{fig:policy-comp-tree} compares the performance of different eviction policies to our MCDP policy.  %, in case of a three-layer edge network, shown in Figure~\ref{fig:cache-7-nodes-non}.  
We plot the relative performance w.r.t. the optimal aggregated utilities of all above policies, normalized to that under MCDP.   We again observe a huge gain of MCDP w.r.t. K-LRU and K-LRU(B) across all values of discount factor. However, the performance gap between MCDP and UB increases with an increase in the value of discount factor.

%i.e., when the variability between content popularities is low.  Performance of both K-LRU and K-LRU(B) increases with the increase of Zipf parameter.}

Finally,  we consider how cache capacity affects the overall objective for both MCDP and UB.  Define $\eta_l = B_l/n, \forall l \in V$ as the ratio of cache size to total number of contents.  For our numerical studies, we take $\eta_l = \eta, \forall l \in V.$  From Figure \ref{fig:capVsutil}, we observe that the aggregate optimal utility increases in $\eta$.  Note that as $\eta$ increases, the aggregate optimal utility gradually converges to a value, and then becomes insensitive to $\eta$.  For a range of simulations performed with different system parameters, we find $\eta \in (0.2,0.3)$ to be high enough to obtain maximum achievable aggregate utility.  This provides design guidelines on how to set cache size for real systems \cite{sundarrajan17, berger17adaptsize, ramadan17,garetto16}.   Furthermore,  the performance gap between UB and MCDP first increases in $\eta$ and later becomes insensitive to $\eta.$

\label{sec:general-cache-duality}

%\section{Non-common Requested Contents}\label{sec:non-common}
%\input{03-non-common}
%
%
%\section{Common Requested Contents}\label{sec:common}
%\input{04-common}
%\input{04A-duality}\label{sec:general-cache-duality}

%\section{Discussion of MCDP policy}\label{sec:disc-mcdp}
%\input{05-discussion-mcdp}

\section{Applications to General Graphs}\label{sec:app}
In this section, we consider a direct application of our framework to content distributions in CDNs, ICNs etc. The problem we consider so far is directly motived by and naturally captures many important realistic networking applications.  These include the Web \cite{che02}, the domain name system (DNS) \cite{jung02dns,mockapetris88},  content distribution networks (CDNs) \cite{maggs15,sitaraman14}, information and content-centric networks (ICNs/CCNs) \cite{jacobson09}, named data networks (NDNs) \cite{jacobson09}, etc.  For example, in modern CDNs with hierarchical topologies, requests for content can be served by intermediate caches placed at edge server that acts as a designated source in the domain.   Similarly, in ICNs/NDNs, named data are stored at designated servers. Requests for named data are routed to the server, which can be stored at intermediate routers to serve future requests.  Both settings directly map to the problem we study here (Section~\ref{sec:common}).   In particular, we consider content distribution in this section. These present hard problems: highly diverse traffic with different content types, such as videos, music and images, require CDNs to cache and deliver content using a shared distributed cache server infrastructure so as to obtain economic benefits.   Our timer-based model with simple cache capacity constraint enable us to provide optimal and distributed algorithms for these applications.

Here we consider a general network topology with overlapping paths and common contents requested along different paths.  Similar to~(\ref{eq:max-mcdp-general-common}) and Algorithm~\ref{algo:primal-dual-alg}, a non-convex optimization problem and a primal-dual algorithm can be formulated and designed. Due to space constraints, we omit the details and only show the performance of this general network. 
 % Instead, we show the performance of this general network. 

We consider two network topologies: Grid and lollipop.  \emph{Grid} is a  two-dimensional square grid while a (a,b)  \emph{lollipop} network is a complete graph of size $a$, connected to a path graph of length $b$.  Denote the network as $G=(V, E)$.  For grid, we consider $|V| = 16$, while we consider a $(3,4)$ lollipop topology with $|V| = 7$ and clique size $3.$  The library contain $|\mathcal{D}|=100$ unique contents.  Each node has access to a subset of contents in the library.  We assign a weight to each edge in $E$, selected uniformly from the interval $[1, 20].$   Next, we generate a set of requests in $G$ as described in \cite{ioannidis16}.  To ensure that paths overlap, we randomly select a subset $\tilde{V}\subset V$ nodes to generate requests.  Each node in $\tilde{V}$ can generate requests for contents in $\mathcal{D}$ following a Zipf distribution with parameter $\alpha=0.8.$  Requests are then routed over the shortest path between the requesting node in $\tilde{V}$ and the node in $V$ that caches the content.   Again, we assume that the aggregate request rate at each node in $\tilde{V}$ is one and the discount factor to be $\psi=0.1.$

\begin{table}[!htbp]
\centering
%\vspace{-0.05in}
\begin{tabular}{c |c |c |c |c}
\hline
\hline
Topology&$\alpha$&MCDP&UB&$\%$ Gap\\
\hline
Grid&$0.8$&$0.0923$&$0.1043$&$11.50$\\
Grid&$1.2$&$0.3611$&$0.4016$&$10.08$\\
Lollipop&$0.8$&$0.0908$&$0.1002$&$9.38$\\
Lollipop&$1.2$&$0.3625$&$0.4024$&$9.91$\\
\hline
\hline
\end{tabular}
\caption{Optimal aggregate utilities under various network topologies.}
\vspace{-0.1in}
\label{tb:topo}
\end{table}

We evaluate the performance of MCDP over the graphs across various Zipf parameter in Table \ref{tb:topo}. It is clear that for both  network topologies, aggregate utility obtained from our TTL-based framework with MCDP policy is higher for higher zipf parameter as compared to  lower zipf parameter. With increase in Zipf parameter, the difference between request rates of popular and less popular contents increases. The aggregate request rate over all contents is the same in both cases. Thus popular contents get longer fraction of rates which in turn yields higher aggregate utility. However, the performance gap between UB and MCDP is around one tenth in both cases and is not affected much by the Zipf parameter.
%\blue{Not a strong conclusion/observation.} 

%
%\vspace{-0.05in}
\section{Minimizing Overall Costs}\label{sec:cost}
In Section~\ref{sec:line}, we focus on maximizing the sum of utilities across all contents over the edge network, which captures user satisfaction. However, communication costs for content transfers across the network are also critical in many network applications. This cost includes (i) \emph{search cost} for finding the requested content in the network; (ii) \emph{fetch cost} to serve the content to the user; and (iii) \emph{transfer cost} for cache inner management due to a cache hit or miss. Below we derive expressions for evaluating these costs for MCDP policy. We relegate the cost analysis for MCD policy to Appendix \ref{appendix:mcd-cost}.

\subsection{Search and Fetch Cost}
A request is sent along a path until it hits a cache that stores the requested content. We define \emph{search cost} (\emph{fetch cost}) as the cost of finding (serving) the requested content in the cache network (to the user). Consider cost as a function $c_s(\cdot)$ ($c_f(\cdot)$) of the hit probabilities. 
 Then the expected search cost across the network is given as
\begin{align}\label{eq:searching-cost}
%S_{\text{MCD}}=
S_{\text{MCDP}}= \sum_{i\in \mathcal{D}} \lambda_i c_s\left(\sum_{l=0}^{|p|}(|p|-l+1)h_{il}\right).
\end{align}
Fetch cost has a similar expression with $c_f(\cdot)$ replacing $c_s(\cdot).$

\subsection{Transfer Cost}
Under TTL, upon a cache hit, the content either transfers to a higher index cache or stays in the current one, and upon a cache miss, the content transfers to a lower index cache (MCDP).  %or is discarded from the network (MCD).  
We define \emph{transfer cost} as the cost due to cache management upon a cache hit or miss.  Consider the cost as a function $c_m(\cdot)$ of the hit probabilities.  

%\noindent{\textit{\textbf{MCD:}}}  
%Under MCD, since the content is discarded from the network once its timer expires, transfer costs are only incurred at each cache hit.  To that end, the requested content either transfers to a higher index cache if it was in cache $l\in\{1,\cdots, |p|-1\}$ or stays in the same cache if it was in cache $|p|.$ Then the expected transfer cost across the network for MCD is given as 
%\begin{align}\label{eq:moving-cost-mcd}
%M_{\text{MCD}} = \sum_{i\in \mathcal{D}} \lambda_ic_m\left(1 - h_{i|p|}\right).
%\end{align} 

%\noindent{\textit{\textbf{MCDP:}}}
{Note that under MCDP, there is a transfer upon a content request or a timer expiry except two cases: (i) content $i$ is in cache $1$ and a timer expiry occurs, which occurs with probability $\pi_{i1}e^{-\lambda_iT_{i1}}$; and (ii) content $i$ is in cache $|p|$ and a cache hit (request) occurs, which occurs with probability $\pi_{i|p|}(1-e^{-\lambda_iT_{i|p|}})$.  Then the transfer cost for content $i$ at steady sate is
\begin{align}
M_{\text{MCDP}}^i&= \lim_{n\rightarrow \infty}M_{\text{MCDP}}^{i, n}=\lim_{n\rightarrow\infty}\frac{\frac{1}{n}M_{\text{MCDP}}^{i, j}}{\frac{1}{n}\sum_{j=1}^n (t_j^i-t_{j-1}^i)}\nonumber\\
&=\frac{1-\pi_{i1}e^{-\lambda_iT_{i1}}-\pi_{i|p|}(1-e^{-\lambda_iT_{i|p|}})}{\sum_{l=0}^{|p|}\pi_l\mathbb{E}[t_j^i-t_{j-1}^i|X_j^i=l]},
\end{align} 
where $M_{\text{MCDP}}^{i, j}$ means there is a transfer cost for content $i$ at the $j$-th request or timer expiry,  $\mathbb{E}[t_j^i-t_{j-1}^i|X_j^i=l]=\frac{1-e^{-\lambda_iT_{il}}}{\lambda_i}$ is the average time content $i$ spends in cache $l.$}

{Therefore, the transfer cost for MCDP is %$M_{\text{MCDP}}=$
\begin{align}\label{eq:moving-cost-mcdp}
&M_{\text{MCDP}}=\sum_{i\in\mathcal{D}}\lambda_iM_{\text{MCDP}}^i\nonumber\displaybreak[0]\\
&= \sum_{i\in\mathcal{D}}\lambda_i\frac{1-\pi_{i1}e^{-\lambda_iT_{i1}}-\pi_{i|p|}(1-e^{-\lambda_iT_{i|p|}})}{\sum_{l=0}^{|p|}\pi_{il}\frac{1-e^{-\lambda_iT_{il}}}{\lambda_i}}\nonumber\displaybreak[1]\\
&=\sum_{i\in\mathcal{D}}\lambda_i\Bigg(\frac{1-\pi_{i1}}{\sum_{l=0}^{|p|}\pi_{il}\mathbb{E}[t_j^i-t_{j-1}^i|X_j^i=l]}+\lambda_i (h_{i1}-h_{i|p|})\Bigg),
\end{align}
where $(\pi_{i0},\cdots, \pi_{i|p|})$ for $i\in\mathcal{D}$ is the stationary distribution for the DTMC $\{X_k^i\}_{k\geq0}$ defined in Section~\ref{sec:ttl-stationary}. %Due to space constraints, we relegate its explicit expression to Appendix \ref{mcdpmodel-appendix}.}

\begin{remark}\label{rem:line-moving-cost}
{The expected transfer cost $M_{\text{MCDP}}$~(\ref{eq:moving-cost-mcdp}) is a function of the timer values. Unlike the problem of maximizing sum of utilities, it is difficult to express $M_{\text{MCDP}}$ as a function of hit probabilities.  }
\end{remark}

\begin{figure}
\centering
\includegraphics[height = 3in, width=3.5in]{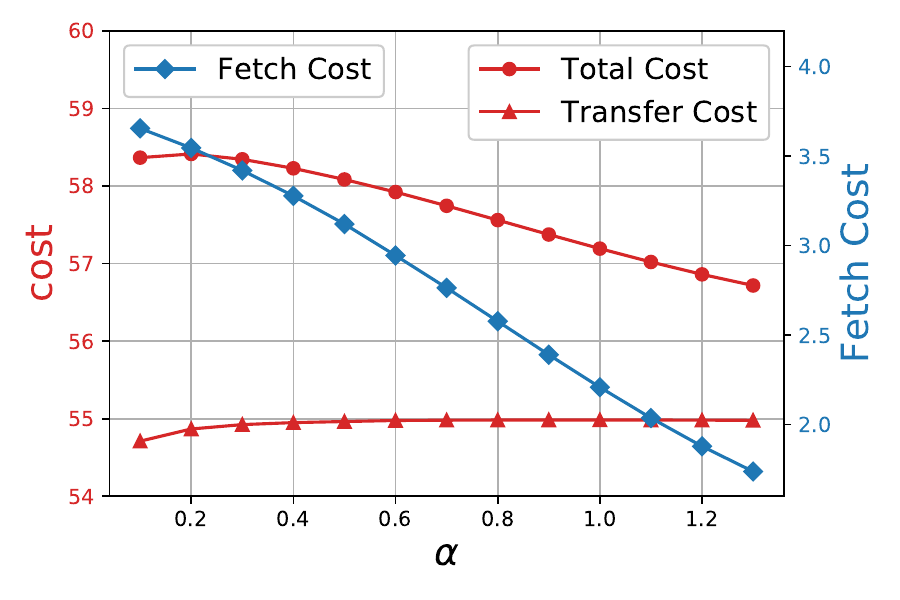}
\caption{Various costs associated with MCDP policy.}
\label{fig:model-linecost}
\vspace{-0.2in}
\end{figure}

\subsection{Variability Improves Performance} 
We study the impact of $\alpha$ on performance for a three node path under the same simulation setting as in Section~\ref{sec:validations-line-cache}.  We solve~(\ref{eq:max-mcdp}) for various values of $\alpha$ and obtain corresponding optimal timers and hit probabilities.  We then compute costs associated with the optimal policy using~(\ref{eq:searching-cost}) and~(\ref{eq:moving-cost-mcdp}), with linear function $c_f$ ($c_s$) with weight $1$ ($0.1$), and plot them in Figure~\ref{fig:model-linecost}.  Note the value of search cost is much smaller than fetch cost, which is consistent with practice, hence we ignore its curve in Figure~\ref{fig:model-linecost}.  As $\alpha$ increases, popular contents get longer fraction of request rates under Zipf distributions. Furthermore,  as $\alpha$ increases, the popular contents are cached closer to the user under an optimal policy (See Section~\ref{sec:validations-line-cache}), hence, the overall search and fetch cost decreases.  However, the transfer cost that constitutes a major proportion of the total cost does not change too much as $\alpha$ increases.  Therefore, the total cost associated with the optimal policy decreases in $\alpha. $

\subsection{Optimization}
Our goal is to determine optimal timer values at each cache on one path in the edge network so that the total costs are minimized. To that end, we formulate the following optimization problem for MCDP
\begin{subequations}\label{eq:min-mcdp1}
\begin{align}
%\text{\bf{L-C-MCDP:}}
\min \quad&{S}_{\text{MCDP}} +{F}_{\text{MCDP}} +{M}_{\text{MCDP}} \displaybreak[0]\\
&\text{Constraints in~(\ref{eq:max-ttl}}).
\end{align}
\end{subequations}

\begin{remark}
As discussed in Remark~\ref{rem:line-moving-cost}, we cannot express transfer cost of MCDP~(\ref{eq:moving-cost-mcdp}) in terms of hit probabilities, hence, we are not able to transform the optimization problem~(\ref{eq:min-mcdp1}) for MCDP into a convex one through a change of variables as we did in Section~\ref{sec:line-utility-max}. Solving the non-convex optimization problem~(\ref{eq:min-mcdp1}) is a subject of future work.  % However, we note that transfer costs of MCD~(\ref{eq:moving-cost-mcd}) are simply a function of hit probabilities and the corresponding optimization problem is convex so long as the cost functions are convex. 
\end{remark}

%\subsubsection{Model Validation and Insights}
%We consider a path with $3$ nodes and set the simulation parameters as described in Section \ref{sec:validations-line-cache}. We solve the non-convex optimization problem using Matlab routine \texttt{fmincon}.  From Figure~\ref{fig:cost}, we observe that popular contents are assigned with higher hit probabilities at cache $3,$ i.e., at the edge cache closest to the user as compared to other caches. The optimal hit probabilities assigned to popular contents at other caches are almost negligible. This advocates for assigning higher probabilities to popular contents and store them close to the user to minimize overall cost.
%
%
%
%
%\begin{figure}
%\centering
%\includegraphics[width=0.8\columnwidth]{figures/n100_alpha08.pdf}
%\caption{Hit probability for MCDP at optimal cost [\np{To be replaced}].}
%\label{fig:cost}
%\vspace{-0.1in}
%\end{figure}
%

\section{Related Work}\label{sec:related}

There is a rich literature on the design, modeling and analysis of cache networks, including TTL caches \cite{rodriguez16,fofack12,fofack14,berger14}, optimal caching \cite{ioannidis16,jian18infocom} and routing policies \cite{ioannidis17}.  In particular, Rodriguez et al. \cite{rodriguez16} analyzed the advantage of pushing content upstream, Berger et al. \cite{berger14} characterized the exactness of TTL policy in a hierarchical topology.  A unified approach to study and compare different caching policies is given in \cite{garetto16} and an optimal placement problem under a heavy-tailed demand has been explored in \cite{ferragut16}. 

Dehghan et al. \cite{dehghan15} as well as Abedini and Shakkottai \cite{abedini14} studied joint routing and content placement with a focus on a bipartite, single-hop setting.  Both showed that minimizing single-hop routing cost can be reduced to solving a linear program. Ioannidis and Yeh \cite{ioannidis17} studied the same problem under a more general setting for arbitrary topologies. 

An adaptive caching policy for a cache network was proposed in \cite{ioannidis16}, where each node makes a decision on which item to cache and evict.  An integer programming problem was formulated by characterizing the content transfer costs.  Both centralized and complex distributed algorithms were designed with performance guarantees.   This work complements our work, as we consider TTL cache and control the optimal cache parameters through timers to maximize the sum of utilities over all contents across the network.  However, \cite{ioannidis16} proposed only approximate algorithms while our timer-based models enable us to design optimal solutions since content occupancy can be modeled as a real variable (e.g. a probability).    

Closer to our work, a utility maximization problem for a single cache was considered under IRM \cite{dehghan16,nitishjian17} and stationary requests \cite{nitishjian-tech17}, while \cite{ferragut16} maximized the hit probabilities under heavy-tailed demands over a single cache.  None of these approaches generalizes to edge networks, which leads to non-convex formulations (See Section~\ref{sec:non-common} and Section~\ref{sec:common}); addressing this lack of convexity in its full generality, for arbitrary network topologies, overlapping paths and request arrival rates, is one of our technical contributions.

\section{Conclusion}\label{sec:conclusion}

We constructed optimal timer-based TTL polices for content placement in edge networks through a unified optimization approach.  We formulated a general utility maximization framework, which is non-convex in general.  We identified the non-convexity issue and proposed efficient distributed algorithm to solve it.  We proved that the distributed algorithms converge to the globally optimal solutions.  We showed the efficiency of these algorithms through numerical studies.  {An analysis of the MCDP algorithm with reset-TTL timers would remain our future work.}

\appendix
\label{appendix}
%\begin{figure*}[htbp]
%\centering
%\begin{minipage}{.33\textwidth}
%\centering
%\includegraphics[width=1\columnwidth]{figures/opt_prob_psipoint4_LRUm.pdf}
%\caption{$\psi=0.4.$}
%\label{fig:line-beta04}
%\end{minipage}\hfill
%\begin{minipage}{.33\textwidth}
%\centering
%\includegraphics[width=1\columnwidth]{figures/opt_prob_psipoint6_LRUm.pdf}
%\caption{$\psi=0.6.$}
%\label{fig:line-beta06}
%\end{minipage}\hfill
%\begin{minipage}{.33\textwidth}
%\centering
%\includegraphics[width=1\columnwidth]{figures/opt_prob_psi1_LRUm.pdf}
%\caption{$\psi=1.$}
%\label{fig:line-beta1}
%\end{minipage}
%\end{figure*}

\subsection{Stationary Behaviors of MCDP}%and MCD}
%\subsubsection{MCDP}
\label{mcdpmodel-appendix}
\cite{gast16} considered a caching policy LRU($\boldsymbol m$). Though the policy differ from MCDP, the stationary analysis is similar.  We present our result here for completeness, which will be used subsequently in the paper.

Under IRM model, the request for content $i$ arrives according a Poisson process with rate $\lambda_i.$ As discussed earlier, for TTL caches, content $i$ spends a deterministic time in a cache if it is not requested, which is independent of all other contents. We denote the timer as $T_{il}$ for content $i$ in cache $l$ on the path $p,$ where $l\in\{1,\cdots, |p|\}.$  

Denote $t_k^i$ as the $k$-th time that content $i$ is either requested or moved from one cache to another.  For simplicity, we assume that content is in cache $0$ (i.e., server) if it is not in the cache network.  Then we can define a discrete time Markov chain (DTMC) $\{X_k^i\}_{k\geq0}$ with $|p|+1$ states, where $X_k^i$ is the cache index that content $i$ is in at time $t_k^i.$  Since the event that the time between two requests for content $i$ exceeds $T_{il}$ happens with probability $e^{-\lambda_i T_{il}},$ then the transition matrix of $\{X_k^i\}_{k\geq0}$ is given as 

{\footnotesize
\begin{align}
{\bf P}_i^{\text{MCDP}}=
 \begin{bmatrix}
    0     &  1 \\
    e^{-\lambda_iT_{i1}}      &  0 & 1-e^{-\lambda_iT_{i1}}  \\
     &\ddots&\ddots&\ddots\\
    &&e^{-\lambda_iT_{i(|p|-1)}}      &  0 & 1-e^{-\lambda_iT_{i(|p|-1)}} \\
    &&&e^{-\lambda_iT_{i|p|}}   & 1-e^{-\lambda_iT_{i|p|}}
\end{bmatrix}.
\end{align}
}
Let $(\pi_{i0},\cdots,\pi_{i|p|})$ be the stationary distribution for ${\bf P}_i^{\text{MCDP}}$, we have
\begin{subequations}\label{eq:stationary-mcdp}
\begin{align}
& \pi_{i0} = \frac{1}{1+\sum_{j=1}^{|p|}e^{\lambda_iT_{ij}}\prod_{s=1}^{j-1}(e^{\lambda_iT_{is}}-1)},\displaybreak[0] \label{eq:stationary-mcdp0}\\
&\pi_{i1} = \pi_{i0}e^{\lambda_iT_{i1}},\displaybreak[1]\label{eq:stationary-mcdp1}\\
&\pi_{il} = \pi_{i0}e^{\lambda_iT_{il}}\prod_{s=1}^{l-1}(e^{\lambda_iT_{is}}-1),\; l = 2,\cdots, |p|.\label{eq:stationary-mcdp2}
\end{align}
\end{subequations}

Then the average time that content $i$ spends in cache $l\in\{1,\cdots, |p|\}$ can be computed as 
\begin{align}\label{eq:averagetime-mcdp}
\mathbb{E}[t_{k+1}^i-t_k^i|X_k^i = l]= \int_{0}^{T_{il}}\left(1-\left[1-e^{-\lambda_it}\right]\right)dt = \frac{1-e^{-\lambda_iT_{il}}}{\lambda_i}, 
\end{align}
and $\mathbb{E}[t_{k+1}^i-t_k^i|X_k^i = 0] = \frac{1}{\lambda_i}.$ 

Given~(\ref{eq:stationary-mcdp}) and~(\ref{eq:averagetime-mcdp}), the timer-average probability that content $i$ is in cache $l\in\{1,\cdots, |p|\}$ is 
\begin{align*}
& h_{i1} = \frac{e^{\lambda_iT_{i1}}-1}{1+\sum_{j=1}^{|p|}(e^{\lambda_iT_{i1}}-1)\cdots (e^{\lambda_iT_{ij}}-1)},\\
& h_{il} = h_{i(l-1)}(e^{\lambda_iT_{il}}-1),\; l = 2,\cdots,|p|,
\end{align*}
where $h_{il}$ is also the hit probability for content $i$ at cache $l.$

\subsection{Stationary Behaviors of Move Copy Down (MCD) policy}\label{mcdmodel-appendix}
MCD behaves the same as MCDP upon a cache hit, i.e. content $i$ is moved to cache $l+1$ preceding cache $l$ in which $i$ is found, and the timer at cache $l+1$ is set to $T_{i{(l+1)}}$.  However, content $i$ is discarded once the timer expires.

We define a DTMC $\{Y_k^i\}_{k\geq0}$ by observing the system at the time that content $i$ is requested. Similarly, if content $i$ is not in the cache network, then it is in cache $0$, thus we still have $|p|+1$ states.  If $Y_k^i=l$, then the next request for content $i$ comes within time $T_{il}$ with probability $1-e^{-\lambda_iT_{il}}$, thus we have $Y_{k+1}^i=l+1,$ otherwise $Y_{k+1}^i=0$ due to the MCD policy. Therefore, the transition matrix of $\{Y_k^i\}_{k\geq0}$ is given as 
{\footnotesize
\begin{align}
{\bf P}_i^{\text{MCD}}=
 \begin{bmatrix}
    e^{-\lambda_iT_{i1}}   & 1-e^{-\lambda_iT_{i1}} & &\\
    e^{-\lambda_iT_{i2}}   & &1-e^{-\lambda_iT_{i2}} & &\\
    \vdots&&\ddots&\\
    e^{-\lambda_iT_{i|p|}}      &  && 1-e^{-\lambda_iT_{i|p|}} \\
    e^{-\lambda_iT_{i|p|}}   && &1-e^{-\lambda_iT_{i|p|}}
\end{bmatrix}.
\end{align}
}
Let $(\tilde{\pi}_{i0},\cdots,\tilde{\pi}_{i|p|})$ be the stationary distribution for ${\bf P}_i^{\text{MCD}}$, then we have 
\begin{subequations}\label{eq:stationary-mcd-app}
\begin{align}
&\tilde{\pi}_{i0}=\frac{1}{1+\sum_{l=1}^{|p|-1}\prod_{j=1}^l(1-e^{-\lambda_i T_{ij}})+e^{\lambda_i T_{i|p|}}\prod_{j=1}^{|p|}(1-e^{-\lambda_i T_{ij}})},\displaybreak[0]\label{eq:mcd1-app}\\
&\tilde{\pi}_{il}=\tilde{\pi}_{i0}\prod_{j=1}^{l}(1-e^{-\lambda_iT_{ij}}),\quad l=1,\cdots, |p|-1,\displaybreak[1]\label{eq:mcd2-app}\\
&\tilde{\pi}_{i|p|}=e^{\lambda_i T_{i|p|}}\tilde{\pi}_{i0}\prod_{j=1}^{|p|-1}(1-e^{-\lambda_iT_{ij}}).\label{eq:mcd3-app}
\end{align}
\end{subequations}

By PASTA property \cite{MeyTwe09}, we immediately have that the stationary probability that content $i$ is in cache $l\in\{1,\cdots, |p|\}$ is given as 
\begin{align*}
h_{il}=\tilde{\pi}_{il}, \quad l=0, 1, \cdots, |p|,
\end{align*}
where $\tilde{\pi}_{il}$ are given in~(\ref{eq:stationary-mcd-app}).
\subsection{Optimization Problem for MCD} \label{app-mcd-opt}
From~(\ref{eq:stationary-mcd-app}), we simply check that there exists a mapping between $(h_{i1},\cdots, h_{i|p|})$ and $(T_{i1},\cdots, T_{i|p|}).$
By simple algebra, we can obtain
\begin{subequations}\label{eq:stationary-mcd-timers}
\begin{align}
& T_{i1} = -\frac{1}{\lambda_i}\log \bigg(1 - \frac{h_{i1}}{1-\big(h_{i1} + h_{i2} + \cdots + h_{i|p|}\big)}\bigg), \label{eq:mcdttl1}\\
& T_{il} = -\frac{1}{\lambda_i}\log\bigg(1-\frac{h_{il}}{h_{i(l-1)}}\bigg),\quad l= 2, \cdots , |p|-1, \label{eq:mcdttl2}\\
& T_{i|p|} = \frac{1}{\lambda_i}\log\bigg(1+\frac{h_{i|p|}}{h_{i(|p|-1)}}\bigg). \label{eq:mcdttl3}
\end{align}
\end{subequations}
Since $T_{il}\geq0,$ we have 
\begin{align} \label{eq:mcd-constraint1}
h_{i(|p|-1)} \leq h_{i(|p|-2)} \leq \cdots \leq h_{i1} \leq h_{i0}.
\end{align}
and 
\begin{align} \label{eq:mcd-constraint2}
h_{i1} + h_{i2} + \cdots + h_{i|p|} \leq 1, 
\end{align}
must hold during the operation.
\subsubsection{Non-common Content Requests under General Edge Networks}\label{appendix:mcd-non-common}
Similarly, we can formulate a utility maximization optimization problem for MCD under general cache network.
\begin{subequations}\label{eq:max-mcd-general}
\begin{align}
\text{\bf{G-N-U-MCD:}} \max \quad&\sum_{i\in \mathcal{D}} \sum_{p\in\mathcal{P}^i}\sum_{l=1}^{|p|} \psi^{|p|-l} U_{ip}(h_{il}^{(p)}) \displaybreak[0]\\
\text{s.t.} \quad&\sum_{i\in \mathcal{D}} \sum_{l\in p}h_{il}^{(p)} \leq B_l,p\in\mathcal{P},   \displaybreak[1]\label{eq:hpbmcd1-general}\\
& \sum_{l=1}^{|p|}h_{il}^{(p)}\leq 1,\quad\forall i \in \mathcal{D}, p\in\mathcal{P}^i,\displaybreak[2]\label{eq:hpbmcd2-general}\\
& h_{i(|p|-1)}^{(p)} \leq \cdots \leq h_{i1}^{(p)} \leq h_{i0}^{(p)}, \quad\forall p\in\mathcal{P}^i, \label{eq:hpbmcd3-general}\\
&0\leq h_{il}^{(p)}\leq1, \quad\forall i \in \mathcal{D}, p\in\mathcal{P}^i.
\end{align}
\end{subequations}

\begin{prop}
Since the feasible sets are convex and the objective function is strictly concave and continuous, the optimization problem defined in~(\ref{eq:max-mcd-general}) under MCD has a unique global optimum.
\end{prop}

\subsubsection{Common Content Requests under General Edge Networks}\label{appendix:mcd-common}
Similarly, we can formulate the following optimization problem for MCD with TTL caches,
\begin{subequations}\label{eq:max-mcd-general}
\begin{align}
&\text{\bf{G-U-MCD:}}\nonumber\\
\max \quad&\sum_{i\in \mathcal{D}} \sum_{p\in \mathcal{P}^i} \sum_{l=1}^{|p|} \psi^{|p|-l} U_{ip}(h_{il}^{(p)}) \displaybreak[0]\\
\text{s.t.} \quad& \sum_{i\in \mathcal{D}} \bigg(1-\prod_{p:j\in\{1,\cdots,|p|\}}(1-h_{ij}^{(p)})\bigg) \leq C_j,\quad\forall j \in V, \displaybreak[1] \label{eq:max-mcd-genenral-cons1}\\
& \sum_{j\in\{1,\cdots,|p|\}}h_{ij}^{(p)}\leq 1,\quad \forall i \in \mathcal{D}, \forall p \in \mathcal{P}^i,   \displaybreak[2]\label{eq:max-mcd-genenral-cons2}\\
& h_{i(|p|-1)}^{(p)} \leq \cdots \leq h_{i1}^{(p)} \leq h_{i0}^{(p)}, \quad\forall p\in\mathcal{P}^i, \displaybreak[3] \label{eq:max-mcd-genenral-cons3}\\
&0\leq h_{il}^{(p)}\leq 1, \quad\forall i\in\mathcal{D}, \forall p\in\mathcal{P}^i. \label{eq:max-mcd-genenral-cons4}
\end{align}
\end{subequations}

\begin{prop}
Since the feasible sets are non-convex, the optimization problem defined in~(\ref{eq:max-mcd-general}) under MCD is a non-convex optimization problem.
\end{prop}

\subsection{Minimizing Overall Costs for MCD}  \label{appendix:mcd-cost}
In the following, we first characterize various costs for MCD.  Then we formulate a minimization problem to characterize the optimal TTL policy for content placement in a path.  

\subsubsection{Search Cost}
Requests from user are sent along a path until it hits a cache that stores the requested content.  We define the \emph{search cost} as the cost for finding the requested content in the cache network. Consider the cost as a function $c_s(\cdot)$ of the hit probabilities.  Then the expected searching cost across the network is given as
\begin{align}\label{eq:search-cost1}
%S_{\text{MCD}}=
S_{\text{MCD}}= \sum_{i\in \mathcal{D}} \lambda_ic_s\left(\sum_{l=0}^{|p|}(|p|-l+1)h_{il}\right).
\end{align}

\subsubsection{Fetch Cost}
Upon a cache hit, the requested content will be sent to the user along the reverse direction of the path. We define the \emph{fetch cost} as the costing of fetching the content to serve the user who sent the request.  Consider the cost as a function $c_f(\cdot)$ of the hit probabilities.  Then the expected fetching cost across the network is given as 
\begin{align}\label{eq:fetching-cost}
%F_{\text{MCD}}=
F_{\text{MCD}}= \sum_{i\in \mathcal{D}} \lambda_ic_f\left(\sum_{l=0}^{|p|}(|p|-l+1)h_{il}\right).
\end{align}

\subsubsection{Transfer Cost}
We define the \emph{transfer cost} as the cost due to caching management upon a cache hit or miss.  Consider the cost as a function $c_m(\cdot)$ of the hit probabilities. Under MCD, since the content is discarded from the network once its timer expires, the transfer cost is caused by a cache hit. To that end, the requested content either moves to a higher index cache if it was in cache $l\in\{1,\cdots, |p|-1\}$ or stays in the same cache if it was in cache $|p|.$ Then the expected transfer cost across the network for MCD is given as 
\begin{align}\label{eq:moving-cost-mcd}
M_{\text{MCD}} = \sum_{i\in \mathcal{D}} \lambda_ic_m\left(1 - h_{i|p|}\right).
\end{align}

\subsubsection{Total Costs} Given the search cost, fetch cost and transfer cost, the total cost for MCD can be defined as
\begin{subequations}\label{eq:total-cost}
\begin{align}
&{SFM}_{\text{MCD}} =S_{\text{MCD}}+F_{\text{MCD}}+M_{\text{MCD}}, \displaybreak[0]\label{eq:total-cost-mcd}
\end{align} 
\end{subequations}
where the corresponding costs are given in~(\ref{eq:search-cost1}),~(\ref{eq:fetching-cost}),~(\ref{eq:moving-cost-mcd}).

%\clearpage
\bibliographystyle{ieeetran}
\bibliography{refs} 

%\section{Appendix}\label{appendix}

%\input{07B-appendix-min-cost}\label{appendix:mcdp-cost}
%\input{07C-appendix-convergence}\label{appendix:convergence}
%\input{07D-appendix-mcd}\label{appendix:mcd}

\end{document}